\RequirePackage{fix-cm}%
\documentclass[smallextended]{svjour3}       % onecolumn (second format)
\smartqed  % flush right qed marks, e.g. at end of proof
\usepackage{graphicx}
\usepackage{xspace}
\usepackage{todonotes}
\usepackage{float}
\usepackage{multicol}
\usepackage{csquotes}
\usepackage{amsmath}
\usepackage{amssymb}
\usepackage{listings}
\usepackage[margin=1.5in]{geometry}    % For margin alignment
\usepackage{algorithmicx}
\usepackage{algcompatible}
\usepackage{algpseudocode, algorithm}
\usepackage{url}
\usepackage{pgfplots}
\usepackage{tikz}
\usepackage{pgfplotstable}
\usepackage{graphicx}
\usepackage{epstopdf}
\usepackage{xcolor}
\usepackage{enumitem}
\usepackage{verbatim}
\usepackage{float}
\usepackage{lineno}
\usepackage{hyperref}
\usepackage{comment}
\usepackage{todonotes}
\usepackage{xspace}
\usepackage{cite}
\usepackage[numbers]{natbib}
\usepackage{titlesec}
\usepackage{float}
\usepackage{blindtext}
\usepackage{enumitem}
\usepackage{amssymb}
\usepackage{pgfplots}
\pgfplotsset{compat=newest}
\DeclareUnicodeCharacter{2212}{-}
\begin{document}
%\newcommand{\todo}[1]{\textcolor{red}{TODO: #1}\PackageWarning{TODO:}{#1!}}

% comment a region
\newcommand{\punt}[1]{}
\newcommand{\cmnt}[1]{}
\algnewcommand{\IIf}[1]{\State\algorithmicif\ #1\ \algorithmicthen}
\algnewcommand{\EndIIf}{\unskip\ \algorithmicend\ \algorithmicif}
% Renews the footnote command
%\renewcommand{\thefootnote}{\alph{footnote}}
% a word should not be broken across lines
\newcommand{\nosplit}{\linebreak}

% no hyphenation
\def\nohyphens{\hyphenpenalty=10000\exhyphenpenalty=10000}

% ~ character
\newcommand{\tilda}{\symbol{126}}

% useful mathematical symbols

\newcommand{\ang}[1]{\langle #1 \rangle}
\newcommand{\Ang}[1]{\Big\langle #1 \Big\rangle}
\newcommand{\ceil}[1]{\lceil #1 \rceil}
\newcommand{\floor}[1]{\lfloor #1 \rfloor}

\newcommand{\cng} {concurrent graph \xspace}
\newcommand{\cgds} {concurrent graph data structure\xspace}
\newcommand{\pset} {partial-set\xspace}

\newcommand{\lble} {linearizable~}
\newcommand{\lbty} {linearizability~}
\newcommand{\rbty} {reachability~}
\newcommand{\legality} {legality\xspace}
%%%%%%%%%%%%%%%%%%%%%%%%%%%%%%%%%%%%%%%%%%%%%%%%%%%%%%%%%%%%%%%%%%%%%%%%%%%%%%%%%%%%%%%%%%%%%%

% keywords for pseudocode

%% various theorem environments
%% Commented out because they are already defined in llncs file %%
%% Decommented out because they are not defined in sigplan class %%
\definecolor{darkblue}{rgb}{0.0, 0.0, 0.55}
\newcommand{\linecomment}[1]{{\scriptsize \textcolor{darkblue}{#1}}}
\newtheorem{thm}{Theorem}
\newtheorem{observation}[thm]{Observation}

\newcounter{history}
\newcommand{\hist}[1]{\refstepcounter{history} {#1}}

\newcommand{\chapref}[1]{Chapter~\ref{chap:#1}}
\newcommand{\secref}[1]{Section~\ref{sec:#1}}
\newcommand{\figref}[1]{Figure~\ref{fig:#1}}
\newcommand{\tabref}[1]{Table~\ref{table:#1}}
\newcommand{\stref}[1]{Step~\ref{step:#1}}
\newcommand{\csref}[1]{Case~\ref{case:#1}}
\newcommand{\thmref}[1]{Theorem~\ref{thm:#1}}
\newcommand{\lemref}[1]{Lemma~\ref{lem:#1}}
\newcommand{\corref}[1]{Corollary~\ref{cor:#1}}
\newcommand{\axmref}[1]{Proposition~\ref{axm:#1}}
\newcommand{\defref}[1]{Definition~\ref{def:#1}}
\newcommand{\eqnref}[1]{Eqn(\ref{eq:#1})}
\newcommand{\eqvref}[1]{Equivalence~(\ref{eqv:#1})}
\newcommand{\ineqref}[1]{Inequality~(\ref{ineq:#1})}
\newcommand{\exref}[1]{Example~\ref{ex:#1}}
\newcommand{\propref}[1]{Property~\ref{prop:#1}}
\newcommand{\obsref}[1]{Observation~\ref{obs:#1}}
\newcommand{\asmref}[1]{Assumption~\ref{asm:#1}}
\newcommand{\thref}[1]{Thread~\ref{th:#1}}
\newcommand{\trnref}[1]{Transaction~\ref{trn:#1}}
\newcommand{\linref}[1]{Line~\ref{lin:#1}}
\newcommand{\algoref}[1]{Algorithm~\ref{algo:#1}}
\newcommand{\subsecref}[1]{SubSection{\ref{subsec:#1}}}

\newcommand{\histref}[1]{\ref{hist:#1}}

\newcommand{\apnref}[1]{Appendix~\ref{apn:#1}}
\newcommand{\invref}[1]{Invariant~\ref{inv:#1}}

\newcommand{\Chapref}[1]{Chapter~\ref{chap:#1}}
\newcommand{\Secref}[1]{Section~\ref{sec:#1}}
\newcommand{\Figref}[1]{Figure~\ref{fig:#1}}
\newcommand{\Tabref}[1]{Table~\ref{tab:#1}}
\newcommand{\Stref}[1]{Step~\ref{step:#1}}
\newcommand{\Thmref}[1]{Theorem~\ref{thm:#1}}
\newcommand{\Lemref}[1]{Lemma~\ref{lem:#1}}
\newcommand{\Corref}[1]{Corollary~\ref{cor:#1}}
\newcommand{\Axmref}[1]{Proposition~\ref{axm:#1}}
\newcommand{\Defref}[1]{Definition~\ref{def:#1}}
\newcommand{\Eqref}[1]{(\ref{eq:#1})}
\newcommand{\Eqvref}[1]{Equivalence~(\ref{eqv:#1})}
\newcommand{\Ineqref}[1]{Inequality~(\ref{ineq:#1})}
\newcommand{\Exref}[1]{Example~\ref{ex:#1}}
\newcommand{\Propref}[1]{Property~\ref{prop:#1}}
\newcommand{\Obsref}[1]{Observation~\ref{obs:#1}}
\newcommand{\Asmref}[1]{Assumption~\ref{asm:#1}}

\newcommand{\Algoref}[1]{Algo~ \ref{algo:#1}}

\newcommand{\Apnref}[1]{Section~\ref{apn:#1}}
\newcommand{\Invref}[1]{Invariant~\ref{inv:#1}}

% environment for writing a proof

%\newcommand{\proof}[1][]{\noindent{\bf\boldmath
%                          P\hspace{-0.25ex}roof{#1}\/:\unboldmath}\hspace*{0.5em}}

\newcommand{\theqed}{$\Box$}

%\newcommand{\qed}{\hspace*{\fill}\theqed\\\vspace*{-0.5em}}

% In case the proof is immediately followed by the start of a new section

\newcommand{\nsqed}{\hspace*{\fill} \theqed}

% Renews the footnote command
\renewcommand{\thefootnote}{\alph{footnote}}

%%%%%%%%%%%%%%%%%%%%%%%%%%%%%%%%%%%%%%%%%%%%%%%%%%%%%%%%%%%%%%%%
% Definitions for new Author Style
%%%%%%%%%%%%%%%%%%%%%%%%%%%%%%%%%%%%%%%%%%%%%%%%%%%%%%%%%%%%%%%%

\newcommand*{\affaddr}[1]{#1} % No op here. Customize it for different styles.
\newcommand*{\affmark}[1][*]{\textsuperscript{#1}}
%\newcommand*{\email}[1]{\texttt{#1}}

%=======================State with no line number=======================================
\def\Statenolinnum#1{{\def\alglinenumber##1{}\State #1}\addtocounter{ALG@line}{-1}}
%==============================================================
%%%%%%%%%%%%%%%%%%%%%%%%%%%%%%%%%%%%%%%%%%%%%%%%%%%%%%%%%%%%%%%%
% Definitions by Petr 
%%%%%%%%%%%%%%%%%%%%%%%%%%%%%%%%%%%%%%%%%%%%%%%%%%%%%%%%%%%%%%%%

\newcommand{\ignore}[1]{}

\newcommand{\para}[1]{\noindent\textbf{\itshape\large#1.}}
\newcommand{\myparagraph}[1]{\noindent\textbf{#1.}}
\algdef{SE}[DOWHILE]{Do}{doWhile}{\algorithmicdo}[1]{\algorithmicwhile\ #1}%
%\algdef{SE}[DOWHILE]{Do}{doWhile}{\algorithmicdo}[1]{\algorithmicwhile\ #1}
%------------------------------------------------------------------------------------------------------------------
% Definitions for Graph Acyclicity Paper
%------------------------------------------------------------------------------------------------------------------
%
\newcommand{\sq}{\hbox{\rlap{$\sqcap$}$\sqcup$}}

\newcommand{\op} {operation\xspace}
\newcommand{\mth} {method\xspace}
\newcommand{\cc} {correctness-criterion\xspace}
\newcommand{\ccs} {correctness-criteria\xspace}
\newcommand{\gen}[1] {gen(#1)}
\newcommand{\term} {term\text{-}op\xspace}
\newcommand{\termop} {terminal operation\xspace}

\newcommand{\eevts}[1] {#1.evts}
\newcommand{\inv} {inv\xspace}
\newcommand{\rsp} {resp\xspace}

\newcommand{\stfdm} {starvation-freedom\xspace}
\newcommand{\stf} {starvation-free\xspace}
\newcommand{\cmth} {commit-throughput\xspace}

\newcommand{\sptree}{{\tt SP\text{-}tree}\xspace}
\newcommand{\vlist} {vertex-list\xspace}
\newcommand{\elist} {edge-list\xspace}
\newcommand{\elists} {edge-lists\xspace}

\newcommand{\helpe}{\textsc{HelpSearchEdge\xspace}}

\newcommand{\of}{obstruction-free\xspace}
\newcommand{\Nbk}{Non-blocking\xspace}
\newcommand{\nbk}{non-blocking\xspace}
\newcommand{\glist}{glist\xspace}
\newcommand{\vcs}{vertices\xspace}
\newcommand{\Vcs}{Vertices\xspace}
\newcommand{\cas}{compare-and-swap\xspace}
\newcommand{\CAS}{\texttt{CAS}\xspace}
\newcommand{\faa}{fetch-and-add\xspace}
\newcommand{\FAA}{\texttt{FAA}\xspace}
%-----------------------------------------------
%\newcommand{\vlist}{
% Definition for pseudocode
%\algrenewcommand{\algorithmiccomment}[1]{$//$ #1}
\newcommand{\pr} {PageRank\xspace}
\newcommand{\ppr} {Previous PageRank\xspace}
\newcommand{\concgraph}{ConcGraph~}
\newcommand{\ds}{data structure~}
\newcommand{\cds}{concurrent data structure~}
\newcommand{\tm}{Transactional Memory~}
\newcommand{\fg}{fine-grained~}
\newcommand{\cg}{coarse-grained~}
\newcommand{\br}{Barriers\xspace}
\newcommand{\bri}{Barriers-Identical\xspace}
\newcommand{\brh}{Barriers-Helper\xspace}
\newcommand{\brc}{Barriers-Chain\xspace}
\newcommand{\be}{Barriers-Edge\xspace}
\newcommand{\bei}{Barriers-Edge-Identical\xspace}
\newcommand{\bec}{Barriers-Edge-Chain\xspace}
\newcommand{\bo}{Barriers-Opt\xspace}
\newcommand{\boi}{Barriers-Opt-Identical\xspace}
\newcommand{\boc}{Barriers-Opt-Chain\xspace}
\newcommand{\beo}{Barriers-Edge-Opt\xspace}
\newcommand{\beoi}{Barriers-Edge-Opt-Identical\xspace}
\newcommand{\beoc}{Barriers-Edge-Opt-Chain\xspace}
\newcommand{\ns}{No-Sync\xspace}
\newcommand{\nsi}{No-Sync-Identical\xspace}
\newcommand{\nsc}{No-Sync-Chain\xspace}
\newcommand{\nse}{No-Sync-Edge\xspace}
\newcommand{\nsei}{No-Sync-Edge-Identical\xspace}
\newcommand{\nsec}{No-Sync-Edge-Chain\xspace}
\newcommand{\nso}{No-Sync-Opt\xspace}
\newcommand{\nsoi}{No-Sync-Opt-Identical\xspace}
\newcommand{\nsoc}{No-Sync-Opt-Chain\xspace}
\newcommand{\nseo}{No-Sync-Edge-Opt\xspace}
\newcommand{\nseoi}{No-Sync-Edge-Opt-Identical\xspace}
\newcommand{\nseoc}{No-Sync-Edge-Opt-Chain\xspace}
\newcommand{\cl}{contributionList\xspace}
\newcommand{\ol}{offsetList\xspace}
\newcommand{\lb}{lock-based~}
\newcommand{\lf}{No-Sync~}
\newcommand{\wf}{Wait-free\xspace}
\newcommand{\df}{deadlock-free\xspace}
\newcommand{\kk}{STICD~}
\newcommand{\kkIN}{STICD-IN~}
\newcommand{\kkCN}{STICD-CN~}
\newcommand{\lfkk}{No-Sync-STICD~}
\newcommand{\lfkkin}{No-Sync-STICD-IN~}
\newcommand{\lfkkcn}{No-Sync-STICD-CN~}
\newcommand{\brkkcn}{Barriers-STICD-CN~}
\newcommand{\brkkin}{Barriers-STICD-IN~}
\newcommand{\node}{node}
\newcommand{\Node}{Node}
\newcommand{\nodes}[1] {#1.nodes\xspace}
\newcommand{\head}{{\tt Head}\xspace}
\newcommand{\vhead}{VHead\xspace}
\newcommand{\ehead}{EHead\xspace}
\newcommand{\tail}{Tail\xspace}
\newcommand{\vtail}{VTail\xspace}
\newcommand{\etail}{ETail\xspace}
\newcommand{\bfsh}{BFSHead\xspace}
\newcommand{\bfst}{BFSTail\xspace}
\newcommand{\headn}{HeadNext\xspace}
\newcommand{\tailn}{TailNext\xspace}
\newcommand{\Init}{\textsc{Init\xspace}}
\newcommand{\hhead}{Head}
\newcommand{\htail}{Tail}
\newcommand{\habs} {AbS\xspace}
\newcommand{\hadd}{HoHAdd\xspace}
\newcommand{\hrem}{HoHRemove\xspace}
\newcommand{\hcon}{HoHContains\xspace}
\newcommand{\hloct}{HoHLocate\xspace}
\newcommand{\hvalid}{HoHValidate\xspace}

\newcommand{\lazy}{lazy-list\xspace}
\newcommand{\hoh}{hoh-locking-list\xspace}

%==========Graph-Applications=================
\newcommand{\putv}{\textsc{PutV}\xspace}
\newcommand{\rmv}{\textsc{RemoveV}\xspace}
\newcommand{\getv}{\textsc{GetV}\xspace}
\newcommand{\pute}{\textsc{PutE}\xspace}
\newcommand{\rme}{\textsc{RemoveE}\xspace}
\newcommand{\gete}{\textsc{GetE}\xspace}
\newcommand{\Put}{\textsc{Put}\xspace}
\newcommand{\Get}{\textsc{Get}\xspace}

\newcommand{\sptclt}{\textsc{SPTClt}\xspace}
\newcommand{\bfstclt}{\textsc{BFSTClt}\xspace}
\newcommand{\dfstclt}{\textsc{DFSTClt}\xspace}
\newcommand{\scctclt}{\textsc{SCCTClt}\xspace}
\newcommand{\bctclt}{\textsc{BCTClt}\xspace}
\newcommand{\nftclt}{\textsc{NFTClt}\xspace}
\newcommand{\diametertclt}{\textsc{DiaTClt}\xspace}
\newcommand{\bfssptclt}{\textsc{SPBFSTreeCollect}\xspace}
\newcommand{\gclt}{\textsc{GraphCollect}\xspace}
\newcommand{\treel}{\textsc{TreeList}\xspace}

\newcommand{\checknegcycle}{\textsc{CheckNegCycle}\xspace}

\newcommand{\getsp}{\textsc{SSSP}\xspace}
\newcommand{\getbfs}{\textsc{BFS}\xspace}
\newcommand{\getdfs}{\textsc{GetDFS}\xspace}
\newcommand{\getscc}{\textsc{GetSCC}\xspace}
\newcommand{\getbc}{\textsc{BC}\xspace}
\newcommand{\getnf}{\textsc{GetNF}\xspace}
\newcommand{\getdia}{\textsc{GD}\xspace}

\newcommand{\spscan}{\textsc{SPScan}\xspace}
\newcommand{\dfsscan}{\textsc{DFSScan}\xspace}
\newcommand{\bfsscan}{\textsc{BFSScan}\xspace}
\newcommand{\sccscan}{\textsc{SCCScan}\xspace}
\newcommand{\bcscan}{\textsc{BCScan}\xspace}
\newcommand{\nfscan}{\textsc{NFScan}\xspace}
\newcommand{\diascan}{\textsc{DiaScan}\xspace}
\newcommand{\bfsspscan}{\textsc{SPBFSScan}\xspace}
\newcommand{\gscan}{\textsc{GraphScan}\xspace}

\newcommand{\csr}{CSR\xspace}

\newcommand \spnote[1] {\todo[inline,size=\footnotesize,color=yellow!20]{Sathya: #1}}

\newcommand \sknote[1] {\todo[inline,size=\footnotesize,color=yellow!20]{Sahith: #1}}

\newcommand \hemanote[1] {\todo[inline,size=\footnotesize,color=yellow!20]{Hema: #1}}

\title{An Improved and Optimized Practical Non-Blocking PageRank Algorithm for Massive Graphs\footnote{This work is already submitted in the 2021 29th Euromicro International Conference on Parallel, Distributed and Network-Based Processing (PDP).}}




\author{Hemalatha Eedi%\footnote{All authors contributed equally and hence listed alphabetically}       
        \and
        Sahith Karra          \and
        Sathya Peri           \and
        Neha Ranabothu        \and        
        Rahul Utkoor\thanks{All authors contributed equally and hence listed alphabetically.}                 
        %etc.
}

\authorrunning{Eedi et. al.} % if too long for running head

\institute{Hemalatha Eedi, \email{cs15resch11002@iith.ac.in} \\
           Sahith Karra, \email{sahith.karra@gmail.com} \\
           Sathya Peri, \email{sathya\_p@cse.iith.ac.in}  \\    
           Neha Ranabothu, \email{cs14btech11028@iith.ac.in} \\  
           Rahul Utkoor, \email{cs14btech11037@iith.ac.in} \\
}

\date{Received: date / Accepted: date}

\titlerunning{An Improved and Optimized Practical Non-Blocking PageRank Algorithm}

\maketitle
%---------------------------------------------------------
\begin{abstract}
PageRank is a well-known algorithm whose robustness helps set a standard benchmark when processing graphs and analytical problems. The PageRank algorithm serves as a standard for many graph analytics and a foundation for extracting graph features and predicting user ratings in recommendation systems. The PageRank algorithm iterates continuously, updating the ranks of the pages till convergence is achieved. Nevertheless, the implementation of the PageRank algorithm on large-scale graphs that on shared memory architecture utilizing fine-grained parallelism is a difficult task at hand. The experimental study and analysis of the Parallel PageRank kernel on large graphs and shared memory architectures using different programming models have been studied extensively. This paper presents the asynchronous execution of the \pr algorithm to leverage the computations on massive graphs, especially on shared memory architectures. We evaluate the performance of our proposed non-blocking algorithms for \pr computation on real-world and synthetic datasets using Posix Multithreaded Library on a 56 core Intel(R) Xeon processor. We observed that our asynchronous implementations achieve 10x to 30x speedup with respect to sequential runs and 5x to 10x improvements over synchronous variants.  

%The experiments are tested on Synthetic and Standard Graph Datasets using Posix Multithreaded Library on a 56 core machine.
%We propose Non-Blocking iterative algorithms which guarantee lock-free and wait-free properties.

\end{abstract}
%--------------------------------------------------------------------------
%---------------------Introduction Starts--------------------------------
\section{Introduction}
\label{intro}
Many practical problems in scientific computing, data analysis, and other thrust areas are modeled as graphs and solved with appropriate graph algorithms\cite{10.5555/301247}.
Most of the graphs are enormous and are scale to billions of nodes and edges while having uncommon and nuanced structures. Processing graphs and graph applications tends to be a performance issue, specifically in shared memory architectures. It is also essential to leverage the existence and interpretation of these large graphs by adding specific metrics for deriving useful analytics on many of these large graphs. PageRank is such a metric that can be used to determine the importance of nodes or pages in a web graph. Page et al. \cite{ilprints422} devised this algorithm for Google Search Engine.  The \pr computation proceeds iteratively to estimate the significance of a web page. In each iteration, we calculate the importance of a page by randomly selecting a page and picking a random link at uniform probability \textit{d} to visit another page. This process continues by updating the rank of a particular page.  Pages with more links are more likely to be visited, so they eventually have higher ranks.  If the outgoing link is not available, then the process moves to a new page with probability \textit{(1-d)} and restarts the process from this page. 

The algorithm's fundamental premise is that a page's rank is determined by its inbound link. Pages with more links are more likely to be visited, resulting in higher rankings.  \cite{ilprints422}. The rank \textit{pr} of node $u$ in Graph G is formally defined as:

\begin{equation}
pr(u)=\frac{1-d}{n}+\ d* \sum_{(v,u) \in E} {\frac{pr(v)}{q}}
\end{equation}
where, \textit{n} = number of pages, \textit{q} = out\-degree defining the number of hyperlinks on page \textit{v} and \textit{d} is the dampening parameter initialized to 0.85.

Parallel implementations of \pr algorithm have been extensively studied on various architectures. As \pr algorithm iteratively progresses, multiple threads coordinate easily using synchronous mechanisms. Synchronization can be applied for both vertex-centric and edge-centric computations and on shared-memory and distributed memory architectures\cite{37252}. The barrier\-synchronization mechanism is more suitable for iterative algorithms such as the \pr algorithm.  However,  synchronous computations utilize Thread-Level Parallelism which leads to drawbacks in dealing with progress conditions in the occurrence of thread failures. On the other hand, in asynchronous computations, progress is guaranteed where threads do not have to wait for slower threads or failure threads. This criterion motivates us to apply asynchronous computations on shared memory architecture for vertex-centric, edge-centric, or graph-centric algorithm implementations. The algorithm implementations relied on processing and computing vertices, in a Vertex-centric model\cite{37252}. Edges are the key computational units in an Edge-centric model\cite{10.1145/2517349.2522740}. In a Graph-centric model, the computations are performed on sub graphs with implicit compiler optimizations. \cite{10.1145/2150976.2151013}

In this paper, we present approximation techniques for our earlier proposed non-blocking methods to leverage the computation of \pr algorithm on massive graphs, especially on shared memory architectures. Our main focus is on designing an asynchronous \pr algorithm with no synchronous limitations that can be applied to vertex-based and edge-based representations. We examined that applying asynchronous computations using the No-Sync variant on the \pr algorithm can speed up performance over synchronous methods. 
The \textit{Loop-Perforation} is an approximate technique that skips some iterations of a loop to increase the computational speed-up. The primary idea of the loop perforation approximate technique is to reduce the amount of computation performed within each iteration as the algorithm makes progress\cite{10.1145/2025113.2025133}\cite{DBLP:conf/hipc/PanyalaSHKCK17}.  
%\subsubsection{Loop-Fusion}
\textit{Loop-Fusion} is an optimization technique that unites two or more independent loops into a single loop and is applied only when data dependencies are preserved. Loop fusion technique when applied increases data locality and the level of parallelism and decreases the overhead of loop control. 
In this direction, we applied loop perforation and loop perforation approximate technique and enabled loop fusion optimization technique to compute the \pr algorithm.
%\spnote{We have to explain what loop perforation and loop fusion are}

%----------------------------------------------------------------------------
\subsection{Our Contributions}
Our contributions in this work are

\begin{itemize}
    \item [$\bullet$] Developed vertex-centric Non-Blocking(\lf and \wf) PageRank algorithms and evaluated the implementations on real-world and synthetic datasets. \cite{DBLP:conf/pdp/EediPRU21}
    \item [$\bullet$] Analysis of vertex-centric and edge-centric computations on PageRank Algorithm. 
    \item [$\bullet$] Applied and analysed  approximation techniques on synchronous and asynchronous algorithms using both vertix-centric and edge-centric computations on PageRank Algorithm. 
    \item [$\bullet$] Testify the performance improvements of asynchronous variants with 5x to 10x speedup when compared with synchronous variants.
\end{itemize} 
%\spnote{Please show the difference between the conference versions and the current paper}
%-------------------Background and Motivation starts-----------
\section{Background and Motivation}
This section presents a description for designing shared data objects using different synchronization techniques
\cite{DBLP:journals/sigsoft/Vu11} for designing shared data objects and algorithms proposed in this paper. Also, we discuss the primary motivation to design an asynchronous technique for the iterative \pr algorithm. In a 
\textbf{Shared Memory Multiprocessors or Multicore systems} multiple processors or processing elements need to coordinate accesses using shared memory. Programming implementations on shared memory systems is challenging as multiple processes simultaneously access shared resources due to lack of coordination, resulting in unpredictable delays and performance bottlenecks. 
An efficient synchronization mechanism is apt to deal with these issues in parallel computation by multiple processes. Two classes of synchronization approaches deal with multiple processes. 

a) According to the \textbf{Blocking synchronization}  mechanism, a shared object can only be accessed by one thread at a time. It locks the other threads and only allows one thread at a time to access the shared object, avoiding conflicts. However, it causes waiting and deadlocks.

b) The \textbf{Non-Blocking synchronization} mechanism is intended to prevent the issues associated with blocking synchronization. \textit{lock-free} and \textit{\wf} methods are used to avoid thread conflicts. The \textit{\wf} technique guarantees that every thread is processed in a finite number of steps, whereas the \textit{lock-free} method guarantees that at least some threads are processed in a finite number of steps. 
The \textit{\wf} method is implemented using the Compare-and-Swap atomic primitive\cite{DBLP:journals/sigsoft/Vu11}. 

%\input{Algorithms/CAS}
%\begin{algorithm}[H]
%\scriptsize
%\caption{Bool Synch CAS(int expected, int updated)}
%\begin{algorithmic}[1]
%\setlength{\lineskip}{2pt}
%procedure computePR(Vertices Thread_vertices, Graph G = (V,E))
%\Procedure{Bool Synch CAS}{int expected, int updated}
%\State $tmp = value$
%\If{$\left value \right == expected$}
%    \State $value = updated$
    
%    \state return true
%\EndIf

%\state return false
%\EndProcedure

%\end{algorithmic}
%\end{algorithm}
\subsection{Motivation}

The research primarily pertains to the graph pre-processing step considering their structural properties and their basic representations \cite{DBLP:conf/ppopp/ShunB13}\cite{DBLP:conf/hipc/PanyalaSHKCK17} \cite{DBLP:conf/icdcn/GargK16}, load distributions on the threads \cite{DBLP:journals/topc/LakhotiaKPP20}, Etc. These algorithms mainly involve the use of the Barrier synchronization concept, but this has many drawbacks. The barrier synchronization applied on each thread at every iteration may cause indefinite blockages without any progress, which stands as a significant issue. Apart from this, parallel algorithmic designs pose many challenges concerning performance and memory. In the context of this, the motive of the algorithm proposed in this paper is to enable independent execution of threads avoiding barriers resulting in increased computational speeds.
Concurrent thread executions are one solution to the parallel algorithm design issues and can help adequately utilize the multi-core architectures of the day. The iterative algorithms use computations of previous iterations for calculating values in the current iterations. When applied to large-scale graphs, Barrier synchronization consumes substantial time and memory resources to yield the desired solution. As is known, the scalability is limited in parallel shared systems due to memory latencies and synchronization problems. 

 It is proposed to parallelize the PageRank algorithm on shared memory systems by improving synchronization aiming fine-grained parallelism using optimistic concurrency control mechanisms and ultimately devise a mechanism to prove the correctness of a parallel algorithm is a challenging task. 

Non-blocking algorithms with lock-free and wait-free properties address these challenges on the PageRank algorithm using piece-wise concurrent programming, removing barrier constraints and iteration dependency from the iterative algorithms considering Graph optimizations as well.

\subsection{System Model for implementing PageRank Algorithm}

One assumption in this paper is that the system has a limited set of \textit{p} threads running on the multiprocessors. The communication between these asynchronous threads happens using shared objects. Atomic primitives such as CAS(Compare-And-Swap) are used to overcome the hurdles encountered during thread communications. To implement a wait-free algorithm, we rely on the CPP vector template library to guarantee thread safety on the lock-free algorithm.

%---------------------Introduction Ends--------------------------------

%----------------------Literature Starts-------------------------------
\section{Related Work}
This section presents an overview of the literature related to the \pr computation from its origin to recent advances. 

Google’s first algorithm proposed by Page et. al \cite{ilprints422} is an iterative algorithm that would rank websites in their search engine results. This concept eventually became the \pr algorithm we know today. The algorithm ranks web pages iteratively until their page rank value converges within a given threshold. These days, the \pr algorithm is a key metric in determining the importance of a given web page. Due to this, its popularity has only grown over the past few decades. The research interest in the \pr algorithm has been prominent in recent history. The parallel computation of the algorithm on shared memory architectures has been the subject of great discussion, with many proposing different programming models for its computation \cite{10.5555/301247},\cite{DBLP:conf/ipps/BarrettBMW09},\cite{DBLP:conf/ppopp/ShunB13},\cite{DBLP:conf/icdcn/GargK16}. 

%\spnote{2009-IPDPS-MTGL}

Berry et al.\cite{DBLP:conf/ipps/BarrettBMW09} proposed a parallel PageRank algorithm in their library known as the Multithreaded Graph Library (MTGL). The algorithm executes on Cray XMT a Multithreaded Scalable Architecture with 128 threads and uses the QThreads API to help process the threads and enable synchronization between them. Each thread computes the PageRank value of a node by accumulating the number of votes of its incoming edges. This implementation had performance bottlenecks due to a lack of optimizations with QThread API. 

%\spnote{2010-UAI-JOURNAL-GraphLab}

GraphLab proposed by Low et al. in their paper \cite{DBLP:conf/uai/LowGKBGH10} is a framework designed to achieve parallelism for Machine learning algorithms. It provides a shared-memory implementation that is efficient but aims to further it to distributed systems and has been used on popular algorithms for testing. The framework has been implemented in C++ with the help of PThreads. The PageRank computations in this framework use synchronization locks and barriers in each iteration and compute the PageRank values using schedulers and assertively tuning needed parameters.  

%\spnote{2013-UAI-CIDR-Grace}

The authors in their paper \cite{DBLP:conf/cidr/WangXDG13} proposed a new programming paradigm – GRACE to facilitate easy programmability along with synchronous and asynchronous executions. The model is built to capture the data dependencies through message passing with faster convergence. The experiments performed on the model show synchronous execution with performance as high as asynchronous executions. The framework provides an iterative synchronous programming approach for developers. A group of worker threads coordinates with a driver thread to calculate the PageRank of the scheduled vertices using Barriers. 

%\spnote{2013-ACM-SIGOPS-Galois}
The authors of \cite{DBLP:conf/sosp/NguyenLP13}, discussed graph mining algorithms with a primary focus on the Page Rank Algorithm. The paper aims to develop a framework for designing scalable data-driven algorithms for graph mining algorithms through a case study on the \pr algorithm. The paper investigates various implementations of the page rank algorithm in the purview of three design axes – work activation, data access pattern, and scheduling criteria to test and understand how various design choices affect the performance. The results showed that considering data-driven designs, which are also scalable over iterative algorithms, improves performance. The results specifically showed that the data-driven, push-back algorithmic implementations had increased the performance by 28x. 

%\spnote{2013-ACM-PPoPP-Ligra}

Shun and Blelloch proposed Ligra - a Lightweight Framework to Process Large Graphs \cite{DBLP:conf/ppopp/ShunB13} is built for shared memory architecture and encompasses both vertex-centric and edge-centric models. As the graphs are stored in the memory, this framework optimizes the computation on shared memory. This framework supports two data types: Graph(V, E), which stores the graph, and another vertexSubset, a subset of |V|. The framework has two functions VERTEXMAP( ) function is used for mapping over the vertices, and EDGEMAP( ) is used for mapping over the edges. Taking advantage of  Frontier Based Computation, an active set of nodes and edges updates dynamically. Ligra uses Cilk Plus parallel codes to achieve parallelism.

%\spnote{2016-ICDCN-STIC-D-KK}

Paritosh Garg and Kishore Kothapally in their paper titled STIC-D: Algorithmic Techniques for Efficient Parallel Pagerank Computation on real-world graphs\cite{DBLP:conf/icdcn/GargK16}, presents four techniques that optimize the large-scaled graphs enough to compute the PageRank algorithm. The first method identifies the Strongly Connected Components(SCCs), performing topological sorting on them. The PageRank is then computed on smaller subgraphs that are strongly connected and processed in a specific order. The second method uses the law that identical nodes have identical PageRank. If two nodes have incoming edges from the same set, then the two nodes would have the same PageRank. So, nodes can be classified, and PageRank is calculated on one vertex from each class. All the other vertices will have the same PageRank, thereby eliminating the redundant computation of similar nodes. If a set of nodes form a chain, each node has only one incoming edge, and one outgoing edge, the PageRank of a vertex with such a node is easy to compute. In the fourth method, if the PageRank does not change in the previous few iterations, we can mark it as a dead node and not include them in the successive iterations. However, the proposed preprocessing methods applied in this paper have not yet been parallelized and require performance enhancements. 

%\spnote{2017-IPDPS-Scott Beamer and David Patterson} 

This paper \cite{DBLP:conf/ipps/BeamerAP17} deals with providing a technique called Propagation Blocking with which the memory communication is reduced. Usually, graphs are sparse, and hence for processing to compute the PageRank, it would take a lot of memory communication. Propagations refer to transfers of vertex values. In the technique proposed, the propagations are stored in memory in a semi-sorted memory. Propagation blocking happens in two phases - binning and accumulate. Binning phase is like bucket sort, where the vertices are divided into different continuous bins. A contribution is processed and is added to the corresponding bin, and the destination is appended to the contribution. In the second phase, accumulate, each bin is processed. Each contribution, destination pair is processed, and the contribution is added to the destination. This approach will reduce the memory communication as the contributions are binned; hence temporally close vertices have more probability of being in a bin. The cache misses will be less, corresponding to less memory communication. Hence his technique improves spatial locality on DRAM by limiting the bounds of memory communication. However, this technique requires an edge-centric representation as input and is bounded by barrier synchronization.

%\spnote{2017-ICCD-GraphRuner}
Hamza Omar et al. in the paper  \cite{DBLP:conf/iccd/OmarAK17}, perform a study on the impact of input dependence for graph algorithms in the context of approximate computing. The authors justify that using perforation on the input graphs over the algorithms improves performance.  Additionally, they proposed a predictor algorithm that helps in reducing the challenges in input dependencies of loop perforation for graphs and enables a satisfactory accuracy level. Experiments were tested using CPU and GPU architectures such as Nvidia, Intel CPU architectures- 8 core Xeon and 61-core Xeon Phi. The results have exhibited a 30\% improvement of performance on using perforation in input graphs and this, when applied to the Nvidia architecture, showed an increase of 19\% of power utilization.

%\spnote{2017-HiPC-Anantha Kalanraman}
In \cite{DBLP:conf/hipc/PanyalaSHKCK17}, the authors aimed to design approximation techniques for computing, enabling good performance coupled with lesser loss of accuracy. The main techniques proposed are loop perforation, vertex/edge ordering, threshold scaling, and other heuristics such as data caching, graph coloring, Etc, which are implemented and tested on the two graph algorithms, i.e., PageRank and community detection. The paper shows the performance improvement of the PageRank algorithm by 83\% and up to 450x for community detection with low influence on the accuracy of using the approximation techniques on the iterative techniques. The authors conclude that approximation techniques will provide good performance with lesser loss of correctness and optimality of solutions. These techniques led to a 7-10 times improvement in performance compared to the efficient algorithms \kk \cite{DBLP:conf/icdcn/GargK16}. Nevertheless, the approximate PageRank computation algorithm utilizes additional storage to save the sorted edge-list for computing the PageRank of the targeted vertex. Barriers synchronization technique is still used for the parallel implementation. 

%\spnote{2018-PACT} 
Graphphi - a framework proposed by \cite{DBLP:conf/IEEEpact/PengPWBR18} with four major components - preprocessing, graph processing model, the MIMD/ SIMD aware scheduler, and extra optimizations. This framework can process edge and vertex-centric graphs, which means this is a hybrid of edge-centric and vertex-centric. This framework works efficiently with MIMD and SIMD with thread-level load balance. The implementation starts with preprocessing graphs into hierarchical blocks. Then the edges are processed, and then these are executed in a load-free, cache-aware, load-balanced, and SIMD-efficient manner. Another optimization layer is provided by push and pull execution and by taking advantage of the High Bandwidth Memory technique. Overall, the framework speed-up by 4X to 35X.   

%\spnote{2018-HPEC} 
An optimized shared-memory graph processing framework introduce by \cite{DBLP:journals/topc/LakhotiaKPP20} increases cache and memory efficiency. This framework is called GPOP (Graph Processing Over Partitions) framework, which promises to increase the efficiency by executing the graph algorithms at lower granularities called partitions. This framework is compared against Ligra, GraphMat, and Galois on different graph algorithms using large datasets to check the efficiency. In comparing the frameworks, GPOP shows fewer cache misses than the other frameworks and increases the performance, which is almost 19x faster than Ligra, 9.3x - GraphMat, and 3.6x – Galois, respectively. 
%The paper discusses GPOP and establishes that the framework improves cache performance, enables faster convergence, and the standard work efficiency of a given graph algorithm. 
%\spnote{2017-HiPC-AK Paper} 
%\spnote{2010-Piccolo} 
%\cite{DBLP:conf/osdi/PowerL10}

%-----------------Literature Ends---------------------------

%--------------Description of Algorithms Starts------------

\section{Description of Algorithms}

This section explains the parallel \pr computation using Blocking and Non-Blocking algorithms on large-scale graphs. 
%---------------------------------------------------
Implementation of iterative parallel graph algorithms takes into consideration the following factors, like convergence, performance. We rely on convergence factor at three different levels for our algorithm: Node-level, Algorithm-level, and Thread-level. In node-level convergence, the termination of the \pr algorithm depends on the convergence of each node independently. In algorithm-level convergence, the termination of the \pr algorithm depends on all nodes from all partitions. In thread-level convergence, each partition terminates independently. In the below algorithms, {\br, \be, \brh} fall under the algorithm-level convergence category, whereas {\ns, \nse} fall under thread-level convergence category. \bo falls under a combination of node-level and algorithm-level convergence categories. \nso fall under a combination of node-level and thread-level convergence categories.

\vspace{-0.6cm}
%%%%%%%%%%%%%%%%%%%%%%%%%%%%%%%%%%%%%%%%%%%%%%%%%%
\subsection{Barrier Algorithm}

The \br algorithm explained here is the baseline version discussed in paper\cite{DBLP:conf/icdcn/GargK16}. Given a graph \textit{G = (V, E)}, vertices are divided into \textit{p} equal-sized partitions. Each thread is responsible for the computation of one partition. We employed a static load allocation technique to assign nodes to partitions. Lines 4 - 9, to begin with, initializes all the variables and the arrays. This algorithm uses two arrays for storing \pr values. The \textit{prev\_pr} array holds the \pr values from the previous iteration, and the \textit{pr} array stores the current iteration \pr values. The error variable helps us decide if the iteration should either continue or converge. This error value is the difference between the \ppr and the \pr for each vertex. The threshold is a constant value initialized to ${10^{-16}}$, which determines the termination condition of the algorithm.

In this algorithm the computation is divided into two phases. Lines 12-18 are the first phase of the algorithm that is responsible for \pr computations. The algorithm computes the maximum absolute difference between the \ppr and the \pr values and saves the value in the thrErr array. Lines 20-22 are responsible for updating the shared variables. After computing the \pr values in the current iteration, the algorithm proceeds to the next phase: to copy the values from the \textit{pr} array to \textit{prev\_pr} array and calculate the global error value.

%-------------Pcode for Baseline Barrier-----------
\begin{algorithm}[H]
\caption{Barrier Baseline Algorithm}
\scriptsize
 %\hspace*{\algorithmicindent} \textbf{Input:} \\
   % \hspace*{\algorithmicindent} \textbf{Output} 
\begin{algorithmic}[1]
\setlength{\lineskip}{3pt}
\State \textbf{Input:} \textit{q} $\leftarrow$ Number of threads
\State  \hspace{0.65cm}  Graph {G} $\leftarrow$ (V, E) \Comment{CSR representation}
%\State uniform partitioning of V into $V_{1}, V_{2} ...,V_{p}$ in increasing order of vertex ids
%\State \textit{m} $\leftarrow$ \textit{maximum} degree of graph
%\frac{10^{-16}}{vertices}
\Procedure{ParallelPageRank}{$G = (V, E), q$}      % \Comment{This is a test}
   % \State System Initialization
   % \State Read the value
    \State  $ err \leftarrow 1 $ \Comment{Initializing Global Error}
%\hemanote{Is it ok to have threshold as 0.01}
    \State  $ threshold \leftarrow \displaystyle{10^{-16}}$
    \ForAll{ nodes $u_{i}$ $\vert i \in  \lbrace 1, ...,\textit{n}\rbrace$} 
        \State $pr(u_{i})\leftarrow 0$
        \State $prPrev(u_{i})\leftarrow \frac{1}{n}$
    \EndFor
    %\While{$error \textless \  threshold$}
    \While{$err > \  threshold$}
        \State $threadErr[q]  \leftarrow 0$ \Comment{Initialize thread's error}
        % \For{all thread $T_{i}$ $\vert i \epsilon  \lbrace 1, ...,\textit{p}\rbrace$} 
   %\State Identify boundary vertices of partition i
   %\State Initialise TotalColors$\left[ m+1 \right] \leftarrow \lbrace 0, 1, .., m \rbrace $
%   \ForAll{ threads $T_{i}$ $\vert i \in  \lbrace 1, ...,\textit{p}\rbrace$}
        \ForAll{ $u \in threadVertices(T_{i}) $ }
            \State $pr(u)\leftarrow \displaystyle\frac{(1-d)}{n} $
            \ForAll{ \textit{u} $\in$ \textit{V}  such  that (v,u) $\in$ E}
                \hspace{2cm} \State $pr(u) = pr(u)+ \displaystyle\frac{prPrev(v)}{outDeg(v)}*d$ 
            \EndFor
            \hspace{1.5cm} \State$threadErr[T_{i}]=max(threadErr[T_{i}],\left| prPrev(u)-pr(u) \right|)$
        \EndFor
%    \EndFor
        \State Barrier Sync Checkpoint
        \Comment{Phase - I}
   %\hemanote{The if condition needs to be written appropriately}
        %\If{$*threadId == 0$}
           % \State $iterations++$
            %\State $error \leftarrow 0$
            \ForAll{ threads $T_{i}$ $\vert i \in  \lbrace 1, ...,\textit{q}\rbrace$} \Comment{Updating Global Error}
                \hspace{1.5cm} \State $err=max(err,threadErr[T_{i}])$  \hspace{1.5cm} 
            \EndFor
            \hspace{1cm} \State $prPrev=pr$
        %\EndIf
        \State Barrier Sync Checkpoint
        \Comment{Phase - II}
  \EndWhile
\EndProcedure

\end{algorithmic}
\label{algorithm:Barriers}
\end{algorithm}
%%%%%%%%%%%%%%%%%%%%%%%%%%%%%%%%%%%%%%%%%%%%%%%%%%%
%\vspace{-3cm}
\subsection{Barrier\_Edge Algorithm} 

\be is the baseline algorithm proposed in \cite{DBLP:conf/hipc/PanyalaSHKCK17} paper. In this approach, the author has developed a three-phase \pr algorithm in which the second and third phases are similar to the \br Algorithm . In the second phase, instead of computing the contribution values of the incoming neighbors, we directly fetch the values from the ContributionList vector. In the first phase, each node traverses through its outLinks and populates the contribution value to its respective outgoing neighbor. Similar to the Barrier Algorithm, each phase ends with barrier instruction.

\begin{algorithm}[H]
\scriptsize
\caption{Barrier-Edge Algorithm}
%\spnote{Barrier_Edge}
\begin{algorithmic}[1]
%Barrier_parellel
\Procedure{Main}{Graph G = (V,E)}
\State G = ConvertCsr(V,E)
\State Initialize Variables
%struct node_info {
%		int start;
%		int end;
%	} inlinks[nodes], outlinks[nodes];
\While{$error > threshold$}
%    \State $maxDiff \leftarrow -\infty$
%    \State parallel \{
    \ForAll{$u \in Thread\_vertices$}
        \If{$outdeg(u) == 0$}
            \State continue
        \EndIf
        %\hemanote{the prev\_pr is the variable the name as per algo in BarrierEdgeOpt. Variable names maybe inconsistent. Same variable maybe used as prprev at last statement of this algo}
        \State $contribution = \displaystyle\frac{prev\_pr(u)}{outdeg(u)}$
        \For{$v \in range(outlinks[u].start, outlinks[u].end)$}
        \Comment{outlinks to u}
            \State $contributionList(offsetList(v)) \leftarrow contribution$
        \EndFor
    \EndFor
%    \State \}
    \State barrier()
    \Comment{Phase - I}
%    \State parallel \{
    \State $threadError[*thdid] \leftarrow 0$
    \ForAll{$u \in Thread\_vertices$}
 %       \State $prev = page\_rank(u)$
        \State $sum \leftarrow 0$
        \For{$v \in range(inlinks[u].start, inlinks[u].end)$}
        \Comment{inlinks to u}
            \State $sum = sum + contributionList(v)$
        \EndFor
        \State $page\_rank(u) = \displaystyle\frac{(1-d)}{n} + (d*sum)$
        \State $thd\_error[thd\_id] = max(thd\_error[thd\_id], prev - page_rank(u))$
    \EndFor
%    \State \}
    \State barrier()
    \Comment{Phase - II}
 %   \State parallel \{
    \If{$thdid == 0$}
%        \State $iterations++$
        \State $error \leftarrow 0$
		\For{$id \in range (0, num\_of\_threads)$}
				\State $error \leftarrow max(error, thd\_error[id])$
        \EndFor
        \State $prprev = pr$
    \EndIf
%	\State \}
	\State barrier()
	\Comment{Phase - III}
\EndWhile
\EndProcedure
\end{algorithmic}
\label{algorithm:Barriers-Edge}
\end{algorithm}
%%%%%%%%%%%%%%%%%%%%%%%%%%%%%%%%%%%%%%%%%%%%%%%%%%%%%%%%
%%%%%%%%%%%%%%%%%%%%%%%%%%%%%%%%%%%%%%%%%%%%%%%%%%%%%
%\subsection{Barrier\_Optimization} 
%\input{Algorithms/Barrier_Opt}
%%%%%%%%%%%%%%%%%%%%%%%%%%%%%%%%%%%%%%%%%%%%%%%%%%%%%
%\subsection{Barrier\_Edge\_Optimization} 
%\input{Algorithms/Barrier_Edge_Opt}
%%%%%%%%%%%%%%%%%%%%%%%%%%%%%%%%%%%%%%%%%%%%%%%%%%%%%

\subsection{No\_Sync}
In our work, we proposed an asynchronous algorithm (\ns) for vertex-centric \pr computations\cite{DBLP:conf/pdp/EediPRU21}. 
At a minimum, one thread must make progress by computing and updating the PageRank values.Multiple threads can access the same element in this approach, but only one thread will be responsible for writing to the memory. In this process, we can encounter read-write conflicts but not write-write conflicts. A thread can read the previous \pr value (or the one computed in the current iteration) in such a scenario. C++ vector templates guarantee this thread-safety property.
\url{https://en.cppreference.com/w/cpp/container}
%(https://en.cppreference.com/w/cpp/container#Thread_safety).

%\hemanote{reference for algorithm 1 needs to be added}
We modified Algorithm \ref{algorithm:Barriers} to make it an asynchronous algorithm. The most notable change is to eliminate \br from the computation at the end of each phase. This change will allow threads to progress in the next iteration without waiting for other threads to complete their task. The subsequent change reduces the memory usage by eliminating the \ppr array. Since we are eliminating the iteration level dependency with our first change, we can apply our second change to Algorithm \ref{algorithm:Barriers}.

Along with the \pr computation, each thread will compute the error value locally. In the synchronous setting, each thread will update the local error value to the global value in the second phase of computation. A thread will update its local error value and partially computed error values from other threads and enter the next iteration in an asynchronous algorithm. This property allows us to have thread-level convergence irrespective of the mode of load allocation.

\begin{algorithm}[H]
\scriptsize
\caption{No-Sync}
\begin{algorithmic}[1]
\setlength{\lineskip}{2pt}
%procedure computePR(Vertices Thread_vertices, Graph G = (V,E))
%\Procedure{computePR}{node u}
%\State return temp
%\EndProcedure
%\vspace{0.05cm}
%\hrule
%\vspace{0.05cm}
\Procedure{Main}{Graph G = (V,E)}
    \State G = ConvertCsr(V,E)
    \State Variables Initialization
%\ForAll{$T_{i} \in threads (1,p) $}
    \While{$localErr > thresholdValue$}
        \State $threadLocalErr \leftarrow 0$
        \ForAll{$ u \in threadVertices(T_{i})$}
            \State $tmp = \displaystyle\frac{(1-d)}{n} $
            \State $previous \leftarrow pr(u)$
            \ForAll{$v \in  V$ such that $(v,u) \in E $}
                \State $tmp = tmp + \displaystyle\frac{pr(v)}{outDegree(v)}  * d $
            \EndFor
            \State $pr(u) \leftarrow tmp$
            \State $threadLocalErr = max(threadLocalErr, \left| tmp - previous \right|) $
        \EndFor
        \State $iterations[*thdid]++$
        \State $threadErr[T_{i}] \leftarrow threadLocalErr $
%\State $localErr \leftarrow 0 $
        \ForAll{$tid \in threads (1,p)$}
            \State $localErr = max(localErr,threadErr[tid]) $
        \EndFor
    \EndWhile
%\EndFor
\EndProcedure
\end{algorithmic}
\label{algorithm:No-Sync}
\end{algorithm}

\maketitle

\begin{lemma}
The algorithm eventually terminates in finite steps.
\end{lemma}

\begin{proof}
The fundamental concept behind this algorithm is that it will terminate when the threshold value is greater than the error value of all the threads. Therefore, it has to be proved that the error value of the threads decreases with each iteration. However, since the maximum error value of the vertices allocated to a thread is the error value of the thread, the statement can be rewritten as follows: it must be proven that in each iteration, the error value of each vertex decreases. The error value in the  \textit{i\textsuperscript{th}} iteration of a vertex \textit{u} and \pr in the baseline algorithm is given by Eq(2) and Eq(3) respectively.

%According to base algorithm, \pr and error of vertex \textit{u} in the \textit{i\textsuperscript{th}} iteration is given by Eq(2) and Eq(3) respectively.

\begin{equation}
    pr_{i}^{u} = \frac{1-d}{n} + \ d*\sum_{(v,u)\in E}^{}\frac{pr_{i-1}^{v}}{outDegree(v)} 
\end{equation}

\begin{equation}
     err_{i}^{u} = \left | pr_{i}^{u} - pr_{i-1}^{u} \right |
\end{equation}

By the definition of our proposed \lf  algorithm, threads simultaneously compute in different iterations. The PageRank of any vertex can be part of any iteration at any given point of time, be in the first or the \textit{n\textsuperscript{th}} iteration. For the base case scenario, threads can be part of two consecutive iterations at a given point of time. 

%In the \lf algorithm, as the threads are allowed to compute in different iterations simultaneously, at a particular instant the PageRank value of a vertex can belong to any iteration (1st to max iteration). As a base case, threads can be considered to be present in two consecutive iterations at a particular instant. 

Eq(2) can be modified to Eq(4) considering that the threads are present in \textit{i\textsuperscript{th}} and \textit{(i-1)\textsuperscript{th}} iterations. Let S\textsubscript{i}\textsuperscript{u} be a set of vertices where (v,u) $\in$ E and PageRank of v is from \textit{i\textsuperscript{th}} iteration. 

\begin{equation}
     pr_{i-1:i}^{u} = \frac{1-d}{n} + \ d*\sum_{v\in S_{i}^{u}}^{}\frac{pr_{i-1}^{v}}{outDegree(v)} + \ d*\sum_{v\in S_{i-1}^{u}}^{}\frac{pr_{i-2}^{v}}{outDegree(v)}
\end{equation}
Error in Eq(3) can also be modified accordingly.
\begin{equation}
    err_{i-1:i}^{u} = \left | pr_{i-1:i}^{u} - pr_{i-1}^{u} \right |
\end{equation}
At any given instant $pr_{i-1}^{u} \leq  pr_{i-1:i}^{u} \leq  pr_{i}^{u}$ if $pr_{i-1}^{u} \leq pr_{i}^{u}$ which means $pr_{i-1:i}^{u}$ always lies between $pr_{i}^{u}$ and $pr_{i-1}^{u}$.
\begin{equation}
    \left | pr_{i-1:i}^{u} - pr_{i-1}^{u} \right | \leq \left | pr_{i}^{u} - pr_{i-1}^{u} \right | \Rightarrow err_{i-1:i}^{u} \leq err_{i}^{u}
\end{equation}
$err_{i}^{u}$ from the base algorithm is always expected to decrease in every iteration, so $err_{i-1:i}^{u}$ also decreases with every iteration. 
\end{proof}

\begin{lemma}
No-Sync algorithm fetches identical results to that of Sequential.
\end{lemma}

\begin{proof}
%\vspace{-0.5cm}
The PageRank of any given vertex is evaluated from the PageRank of all of its incoming vertices. Since the threads can determine PageRank simultaneously in a different iteration, the PageRank of any given vertex is evaluated from the PageRank of incoming vertices that could be part of any iteration.
%PageRank of a vertex is computed from the PageRank of all its incoming vertices. As the threads are allowed to compute in different iterations simultaneously, the PageRank of a vertex can be computed from the PageRank of incoming vertices which may belong to any iteration.
Eq(4) can be modified for the threads to be present in 1\textsuperscript{st} to i\textsuperscript{th} iteration.

\begin{equation}
    \widehat{pr_{i}^{u}} = \frac{1-d}{n} + \ d*\sum_{l=1}^{I}\sum_{v\in S_{l}^{u}}^{}\frac{\widehat{pr_{l}^{v}}}{outDegree(v)}
\end{equation}

The algorithm continues until the error of every node is less than the threshold, so the PageRank values of all nodes reach an almost constant value. With the given termination condition the Eq(7) can be modified as Eq(8) where $S_{}^{u} =\bigcup\limits_{l=1}^{I} S_{l}^{u} = \{v|(v,u)\in E\} $.
\begin{equation}
    \widehat{pr_{}^{u}} = \frac{1-d}{n} + \ d*\sum_{v\in S_{}^{u}}^{}\frac{\widehat{pr_{}^{v}}}{outDegree(v)}
\end{equation}

As per the defined termination condition, the acquired error value from the modified PageRank values is less than the threshold. Therefore, the PageRank values from the No-Sync algorithm are identical to Sequential, with the given error value less than the allocated threshold value. 

%The error obtained from the modified PageRank values is less than the threshold based on the termination condition. Hence, the PageRank values from the algorithm are also similar to that of the Sequential output with an error which is less than the threshold.

Eq(8) is exactly same as Eq(2) where $|pr^{u} - \widehat{pr^{u}}| \leq threshold$ is satisfied only at the termination condition. This Lemma is also proved experimentally and the L1 norm of the PageRank values is less than 1/10th of the threshold for all the experiments.
\end{proof}
%\vspace{-2cm}
%%%%%%%%%%%%%%%%%%%%%%%%%%%%%%%%%%%%%%%%%%%%%%%%%%%%%%
\subsection{No\_Sync\_Edge}
Likewise to how we developed an asynchronous algorithm for a 2-phased \pr computational model, we also developed an asynchronous algorithm for the 3-phased \pr computational model.
The changes proposed in the previous algorithm are also applicable for this variant. In  Algorithm \ref{algorithm:Barriers}, \pr computations are happening in one single equation, whereas in Algorithm \ref{algorithm:Barriers-Edge}\be, we split the equation into two parts. Though we successfully developed asynchronous variants for both algorithms, this variant does not guarantee convergence for particular types of datasets. This variant resulted in better speedups when we tested it on our synthetic datasets; however, it did not converge with the given threshold for standard datasets. Since the asynchronicity is entirely random, we are still exploring the reasons behind the non-convergence of this variant.

\begin{algorithm}[H]
\scriptsize
\caption{No-sync-Edge}
\begin{algorithmic}[1]
%No-Sync Barrier_parallel
\Procedure{Main}{Graph G = (V,E)}
\State G = ConvertCsr(V,E)
%\State Initialize Variables
%struct node_info {
%		int start;
%		int end;
%	} inlinks[nodes], outlinks[nodes];
    \While {$error > threshold$}
%	\State $maxDiff \leftarrow -\infty$
%	\State parallel \{
        \State $thd\_error[*thdid] \leftarrow 0$
		\ForAll {$u \in Thread\_vertices$}
			\State $prev \leftarrow page_rank(u)$
			\State $sum \leftarrow 0$
			\For {$v \in range(inlinks[u].start, inlinks[u].end)$} \Comment{inlinks to u}
				\State $sum = sum + contributionList(v)$
			\EndFor
		    \State $page\_rank(u) = \displaystyle\frac{(1-d)}{n} + (d*sum)$
			\State $thd\_error[thd\_id] = max(thd\_error[thd\_id], prev - page\_rank(u))$
		\EndFor
		\State $interation[*thdid]++$
		\State $error \leftarrow 0$
		\For($id \in range (0, num\_of\_threads)$)
			\State $error \leftarrow max(error, thd\_error[id])$
        \EndFor
		\ForAll($u \in Thread\_vertices$)
		    \If{$outdeg(u) == 0$}
		        \State continue
		    \EndIf
			\State $contribution = \displaystyle\frac{page\_rank(u)}{outdeg(u)}$
			\For{$v \in range(outlinks[u].start, outlinks[u].end)$} \Comment{outlinks to u}
				\State $contributionList(offsetList(v)) \leftarrow contribution$
			\EndFor
		\EndFor
%	\State \}
    \EndWhile
\EndProcedure
\end{algorithmic}
\label{algorithm:No-Sync-Edge}
\end{algorithm}
%%%%%%%%%%%%%%%%%%%%%%%%%%%%%%%%%%%%%%%%%%%%%%%%%%%%%%
\vspace{-2cm}
\subsection{Barrier and No\_Sync\ Variants Optimization}
Many applications might not require the exact solution which can help reduce the overall computational cost. When using approximation techniques, we skip some portions of the computation to arrive at an approximate solution. This technique can significantly improve the performance by a minimum compromise on accuracy.

Loop perforation is an optimization technique that can reduce iterations without changing the definite description of an algorithm when applied to iterative algorithms. We used this technique for our \br and \ns variants as proposed in paper\cite{DBLP:conf/hipc/PanyalaSHKCK17} to compute the \pr algorithm. We made a slight modification to the author's technique, where we are eliminating the \pr computations if the absolute difference of \pr and \ppr of a node is less than ${10^{-21}}$.

%Loop perforation is an optimization technique that when applied to iterative algorithms can reduce the number of iterations without changing the definite  description of an algorithm. We applied this technique for or \br and \ns variants as proposed in paper\cite{DBLP:conf/hipc/PanyalaSHKCK17} to compute the \pr algorithm. As shown in line 11 of Algorithm 5, the optimization is performed in verifying the absolute values of \pr and \ppr to the threshold limit with value ${10^{-21}}$.
%\input{Algorithms/BrNsOptAlgos} 
%\textbf{pcode for Barrier Optimization Version}
%Convergence Logic is inserted
\begin{algorithm}[H]
\caption{Barrier Optimization and No\-Sync Optimization Algorithm}
%\hemanote{Applying Convergence Logic + Loop Perforation}
\scriptsize
 %\hspace*{\algorithmicindent} \textbf{Input:} \\
   % \hspace*{\algorithmicindent} \textbf{Output} 
\begin{algorithmic}[1]
\setlength{\lineskip}{3pt}
    %\While{$error \textless \  threshold$}
                 \vdots \\
    \While{$error > \  threshold$}
        % \For{all thread $T_{i}$ $\vert i \epsilon  \lbrace 1, ...,\textit{p}\rbrace$} 
   %\State Identify boundary vertices of partition i
   %\State Initialise TotalColors$\left[ m+1 \right] \leftarrow \lbrace 0, 1, .., m \rbrace $
        \State $threadErr[q]  \leftarrow 0$ \Comment{Initialize thread's error}
%   \ForAll{ threads $T_{i}$ $\vert i \in  \lbrace 1, ...,\textit{p}\rbrace$}
%\hemanote{Variable names in if conditions and statements needs to be verified}
        \ForAll{ $u \in threadVertices(T_{i}) $ }
            \State $pr(u)\leftarrow \displaystyle\frac{(1-d)}{n} $
            \If{$threshold\_check[i] == False$}
                \ForAll{ \textit{u} $\in$ \textit{V}  such  that (v,u) $\in$ E}
                    %\hspace{2cm} \State $pr(u) = pr(u)+ \displaystyle\frac{prPrev(v)}{outDeg(v)}*d$ 
                    \hspace{2cm} \State $pr(u) = pr(u)+ contribution(v) * d$ 
                \EndFor
                \hspace{1.5cm} \State$threadErr[T_{i}]=max(threadErr[T_{i}],\left| prPrev(u)-pr(u) \right|)$
                \If{$\left| prPrev(u)-pr(u) \right| !=0  \&\& \left| prPrev(u)-pr(u) \right|<threshold*0.00001$}
                    \State $threshold\_check[i] = True$
                \EndIf
            \EndIf
        \EndFor
%    \EndFor
    \EndWhile
    
    \vdots
\end{algorithmic}
\label{algorithm:Barriers-No-Sync-Opt}
\end{algorithm}

%%%%%%%%%%%%%%%%%%%%%%%%%%%%%%%%%%%%%%%%%%%%%%%%%%%%%
\subsection{Barrier\_Helper} 

Barrier helper is a wait-free algorithm centered around solving thread delay and failure issues, thus ensuring an algorithm's correctness. In each iteration in the barrier helper algorithm, threads may not enter the next iteration unless the PageRank value of each node is calculated for the given iteration. The main motive behind this Barrier-Helper is that threads completed with their task are re-assigned to other threads to help them compute the PageRank for that iteration, thus actively avoiding failure and delay scenarios. This process continues until the PageRank value has been calculated for each node in that iteration. 
Algorithm-\ref{algorithm:Barriers-Helper} explains the flow for PageRank computation of vertices using the wait-free approach. \enquote{Please refer to our conference paper\cite{DBLP:conf/pdp/EediPRU21} for details}. 
%-----------------------------------------

\begin{lstlisting}[ ]
struct ThreadCASObj{
    int itrNum;
    int currNode;
    double thErr;
};
struct GlobalCASObj{
    int itrNum;
    double err;
    vector<bool>check;
    bool intermediate;
};
struct PrCASObj{
    int itrNum;
    double rank;
};
\end{lstlisting}

\begin{algorithm}[!ht]
\scriptsize
\caption{Wait-Free}
\begin{algorithmic}[1]

\Procedure{updatePageRank}{u, vertexPr, threadVar}
	\State $tmp \leftarrow pr(u) $
	\If{$tmp.itrNum == threadVar.itrNum$}
		\State $casObj \leftarrow new PrCASObj(threadVar.itrNum++, vertexPr)$
		\State $CAS(pr(u), tmp, casObj)$
	\EndIf

	\State $tmp \leftarrow prevPr(u) $
	\If{$tmp.itrNum == threadVar.itrNum$}
		\State $casObj \leftarrow new PrCASObj(threadVar.itrNum++, vertexPr)$
		\State $CAS(prevPr[u], tmp, casObj)$
	\EndIf
\EndProcedure
\vspace{0.05cm}
\hrule
\vspace{0.05cm}

\Procedure{computePR}{threadId, helpId, threadVar}
	\State $threadInfo \leftarrow  globalThreadInfo[helpId].load()$

	\While{$u \in ThreadVertices$ and $globalVar.itrNum == threadVar.itrNum$}
		
		\State $vertexPr \leftarrow \frac{(1-d)}{n}$
		\ForAll{$ v \in V$ such that $(v, u) \in E$}
     		\State $vertexPr += \displaystyle\frac{globalPrevPr[v]}{outDeg(v)} * d $
   		\EndFor

   		\State Invoke updatePageRank(u,vertexPr,threadVar)

   		\State $tmp = globalThreadInfo[helpId]$
   		\If{$tmp.itrNum == threadVar.itrNum$}
   			\State $er \leftarrow max(tmp.er,\left | vertexPr - prevPr \right | )$
   			\State $casObj \leftarrow new ThreadCASObj(tmp.itrNum,next(u, hepId), er) $
   			\State $CAS(globalThreadInfo[helpId],tmp,casObj) $
   		\EndIf
		\EndWhile
\EndProcedure
\vspace{0.05cm}
\hrule
\vspace{0.05cm}

\Procedure{UpdateGlobalVariable}{(thId,helpId,threadVar}

	\While{$true$} 
		\State $tmp \leftarrow globalVar$
		\If{$tmp.itrNum == threadVar.itrNum$}
			\State $casObj \leftarrow copy(tmp)$
			\State $casObj.check[helpId] \leftarrow true$
			\State $casObj.er \leftarrow max(casObj.er,globalThreadInfo[helpId].er)$
		 	\State $CAS(globalVar,tmp,casObj) break $
		\EndIf
	\EndWhile

	\While{$true$}
		\State $tmp \leftarrow globalThreadInfo[helpId]$
		\If{$tmp.itrNum == threadVar.itrNum$}
			\State $casObj \leftarrow new ThreadCASObj(tmp.itrNum++,threadId,0)$
			\State $CAS(globalThreadInfo[helpId],tmp,casObj) break$
		\EndIf
	\EndWhile
	
\EndProcedure
\vspace{0.05cm}
\hrule
\vspace{0.05cm}

\Procedure{computeThreadPageRank}{threadId}
	\While{$globalVar.load().er > threshold$}

		\State Invoke computePR(threadId,threadId,threadVar)
		\ForAll{$ thr \in threads$ and $thr != threadId$ and notCompletePR(thr)}
			\State Invoke computePR(thr,threadId,threadVar)
		\EndFor
		\State Intialize error value to 0 for next iteration in globalVar using CAS
		\State Invoke UpdateGlobalVariable(threadId,threadId,threadVar)
		\ForAll{$thr \in threads$ and $thr != threadId$ and notCompleteglobalVar(thr)}
			\State Invoke UpdateGlobalVariable(thr,threadId,threadVar)
		\EndFor

		\State Intialize itrNum in globalVar using CAS

		\State $threadVar \leftarrow globalVar.load()$

		\State Execute SwapFun() 
\EndWhile
\EndProcedure

\end{algorithmic}
\label{algorithm:Barriers-Helper}
\end{algorithm}
%\input{Algorithms/Classes}
\iffalse
\newline
\newline
\textbf{struct} ThreadCASObj \{\\
\hspace*{1cm} int itrNum;
\newline \hspace*{1cm} int currNode;
\newline \hspace*{1cm} double thErr;
\newline  \}
\newline
\textbf{struct} GlobalCASObj \{
\newline \hspace*{1cm} int itrNum;
\newline \hspace*{1cm}  double err;
\newline \hspace*{1cm}  vector$<$bool$>$ check;
\newline \hspace*{1cm}  bool intermediate;
\newline  \};
\newline
\textbf{struct} PrCasOb \{
\newline \hspace*{1cm}  int itrNum;
\newline \hspace*{1cm}  double rank;
\newline \};
\newline
\fi

%%%%%%%%%%%%%%%%%%%%%%%%%%%%%%%%%%%%%%%%%%%%%%%%%%%%%%%%%%

%\subsection{No\_Sync\_Opt}
%\input{Algorithms/No-Sync_Opt}

%%%%%%%%%%%%%%%%%%%%%%%%%%%%%%%%%%%%%%%%%%%%%%%%%%%%%%

%\subsection{No\_Sync\_Edge\_Opt}
%\input{Algorithms/No-Sync_Edge_Opt}

%%%%%%%%%%%%%%%%%%%%%%%%%%%%%%%%%%%%%%%%%%%%%%%%%%%%%%
\clearpage
%---------Description of Algorithms Ends-------------

%-----------Experiments  & Results Starts----------
\section{Experiments Evaluation}
\subsection{Platform}
%We conducted our simulations on a 56 core Intel(R) Xeon(R) E5-2660 v4 processor architecture running at 2.06 GHz core frequency. Each core supports two logical threads and two CPU socket(s) with 14 cores per socket. Every core’s L1 - 32K, L2 - 256K cache memory is private to that core, and L3 - 35840K cache memory is shared across the cores. All the simulations were coded in C/C++ and compiled using g++ 7.5.0 and using the POSIX Multi-Threaded library.

The simulations were conducted on a 56 core Intel® Xeon® E5-2660 v4 processor that runs at 2.06 GHz core frequency. With this architecture, the two CPU sockets support 14 cores each, with each core supporting up to two logical threads. Also, every core has an L1-32K, L2 – 256K cache memory specific to the core, and L3 – 35840K cache memory. Implementations are coded in C/C++, and compilation was performed using g++ 7.5.0 and POSIX Multithreaded library. 

\subsection{Datasets}

The experiment uses synthetic and real-world datasets for the simulations. There are three categories of real-world datasets taken from SNAP repository \cite{snapnets} and randomly generated synthetic datasets . All the datasets have been considered keeping the studies \cite{DBLP:conf/ppopp/ShunB13},\cite{DBLP:conf/icdcn/GargK16},\cite{DBLP:journals/fgcs/LuoL20} in purview for providing a comparison. The datasets mentioned above are shown in Table 1. The initial experimentation was conducted on synthetic graphs generates randomly in the range $1*10^6 to 7*10^6$ using the RMAT graph library\cite{DBLP:conf/sdm/ChakrabartiZF04}. Further experiments were conducted on the standard datasets repository – Web-Graphs, Social-Networks, and Road-Networks. The formats of the graph dataset sizes are in Adjacency List \cite{DBLP:journals/fgcs/LuoL20} and are converted later to Compressed Sparse Row(CSR) format, and all the codes are tested with the given datasets.  
%\hemanote{Graph Representations}
%\input{Table_Dataset}

\begin{table}[!htbp]
%% increase table row spacing, adjust to taste
\renewcommand{\arraystretch}{1.3}
% if using array.sty, it might be a good idea to tweak the value of
%\extrarowheight as needed to properly center the text within the cells 5

\caption{Real-world and Synthetic Graph Datasets}
\label{table:datasets}
\centering
%% Some packages, such as MDW tools, offer better commands for making tables
%% than the plain LaTeX2e tabular which is used here.
\begin{tabular}{|c|c|c|c|}
%\hline
\textbf{Input} & \textbf{\# of vertices} & \textbf{\# of Edges} & \textbf{Size in MB} \\
\hline
\multicolumn{4}{|c|}{ \textbf{\textit{ Web Graphs \cite{snapnets}}}}\\
\hline
webStanford & 281903 & 2312497 & 30\\
\hline
webNotreDame & 325729 & 1497134 & 20\\
\hline
webBerkStan & 685230 & 7600595 & 20\\
\hline
webGoogle & 875713 & 5105039 & 7\\
\hline
\multicolumn{4}{|c|}{ \textbf{\textit{ Social Networks \cite{nr}}}}\\
\hline
socEpinions1 & 75879 & 508837 & 5.7\\
\hline
Slashdot0811 & 77360 & 905468 & 10.7\\
\hline
Slashdot0902 & 82168 & 948464 & 11.3\\
\hline
socLiveJournal1 & 4847571 & 68993773 & 1100\\
\hline
\multicolumn{4}{|c|}{ \textbf{\textit{ Road Networks\cite{nr}}}}\\
\hline
roaditalyosm & 6686493 & 7013978 & 109.9\\
\hline
greatbritainosm & 7700000 & 8200000 & 28\\
\hline
asiaosm & 12000000 & 12700000 & 5.1\\
\hline
germanyosm & 11500000M & 12400000 & 98.5\\
\hline
\multicolumn{4}{|c|}{\textbf{\textit{ Synthetic Graphs \cite{DBLP:conf/sdm/ChakrabartiZF04} \cite{DBLP:journals/fgcs/LuoL20} }}}\\
\hline
D10 & 491550 & 0999999 & 13.2\\
\hline
D20 & 954225 & 1999999 & 28.3\\
\hline
D30 & 1400539 & 2999999 & 43.3\\
\hline
D40 & 1871477 & 3999999 & 59.0\\
\hline
D50 & 2303074 & 4999999 & 74.1\\
\hline
D60 & 2759417 & 5999999 & 89.9\\
\hline
D70 & 3222209 & 6999999 & 105.6\\
\hline
\end{tabular}
\end{table}
------------------------------------------------------------------
%\begin{table}[!t]
%% increase table row spacing, adjust to taste
%\renewcommand{\arraystretch}{1.3}
% if using array.sty, it might be a good idea to tweak the value of
% \extrarowheight as needed to properly center the text within the cells
%\caption{An Example of a Table}
%\label{table_example}
%\centering
%% Some packages, such as MDW tools, offer better commands for making tables
%% than the plain LaTeX2e tabular which is used here.
%\begin{tabular}{|c||c|}
%\hline
%One & Two\\
%\hline
%Three & Four\\
%\hline
%\end{tabular}
%\end{table}

\subsection{Results and Discussion}
The section presents the speed-up obtained with the parallel variants of the PageRank algorithm. The algorithm's speed-up is calculated using the ratio of Sequential execution time vs. Parallel execution time. The programs are executed with a fixed number of threads(56) on different classes of datasets in order to obtain the execution times. The proposed algorithms have shown significant improvement at the hardware level by incorporating them alongside current graphs processing techniques.\\
%----------------------------------------------------------

%----------------pgfPlot1---------------------
\begin{figure}[H]
\vspace*{-3.5cm}
%\begin{adjustbox}{width=110mm, height=100mm}
\centering
\includegraphics[width=6in]{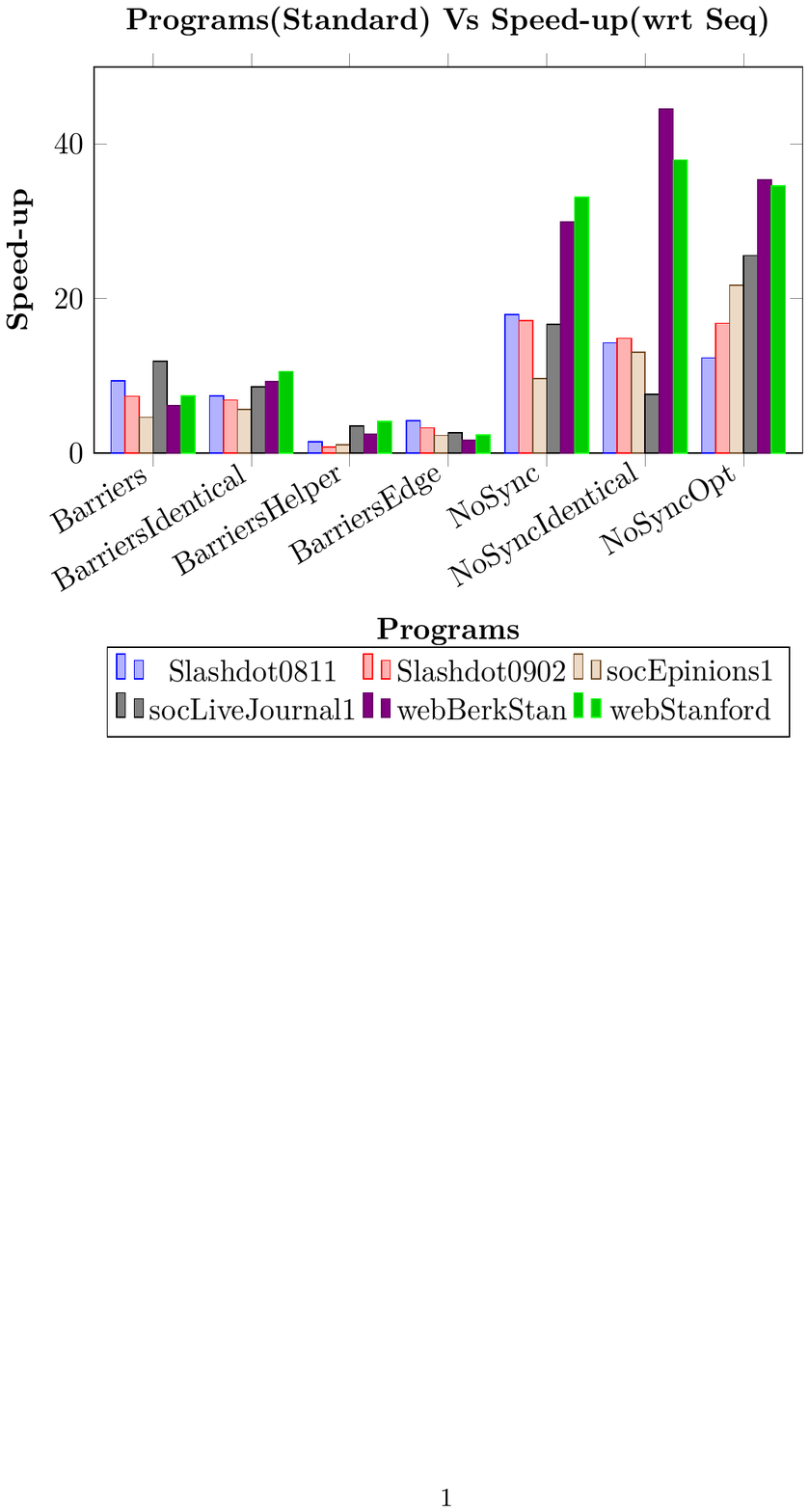}
\DeclareGraphicsExtensions.
\vspace*{-10cm}
\caption{Speed-Up Vs Programs on Standard Datasets (56 Threads)}
\label{fig:lableFig1}
%\end{adjustbox}
\end{figure}

%\hemanote{Description for Plot1}
\figref{lableFig1} Shows the speedups obtained by parallel variants(blocking and non-blocking variants) on standard datasets using 56 threads.  Barrier variants result in a maximum of 10x on standard datasets, whereas \ns variants (except for \nse) consistently produce greater than 10x speedup on almost all datasets. It is observed from the results that \ns, \nsi, \nso and \nsoi, are giving better performance than the \br, \bri and \be on all the datasets. We achieve substantial performance benefits by removing the barriers and allowing partial computations on shared variables to eliminate iteration-level dependency and thread-level dependency. Thus it can be concluded that asynchronous variants outperform synchronous variants by a considerable magnitude. As each thread progresses independently and completes the given task, we achieve the lock-free property on the \ns variants. We conclude that the lock-free variants of the PageRank algorithm provide better performance improvements compared to the other variants. The notion behind the \wf variant is to display the sustainability of the current program execution and hence is not explicitly designed for performance. Since we are not using any compiler optimization flags, the \be variant is not as performant. 

%----------------------------------------------------------
%\hemanote{Description for Plot2}
\figref{lableFig2} shows the speedups obtained by parallel variants on synthetic datasets. The insights noted in \figref{lableFig1} are also applicable here for synthetic datasets.  Barrier variants result in a maximum of 5x speedup on synthetic datasets, whereas No-Sync variants (except for \nse) consistently produce greater than 10x speedup on almost all datasets. It is observed for Synthetic datasets that as size increases, \ns variants consistently outperform \br variants in terms of performance.
%---------------pgfplot2-----------------------
\begin{figure}[!ht]
\vspace*{-3cm}
%\begin{adjustbox}{width=110mm, height=100mm}
\centering
\includegraphics[width=6in]{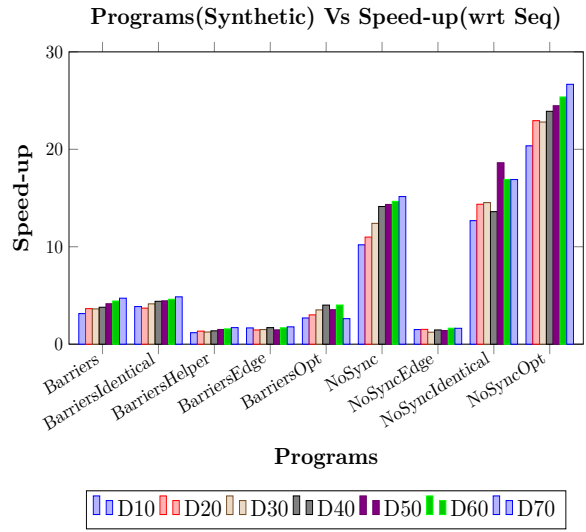}
\DeclareGraphicsExtensions.
\vspace*{-10cm}
\caption{Speed-Up Vs Programs on Synthetic Datasets (56 Threads)}
\label{fig:lableFig2}
\end{figure}
%--------------------------------------------------
%\hemanote{Description for Plot3 and Plot4}
\\
\figref{lableFig3} and \figref{lableFig4} shows the speedups gained by the parallel version by varying thread count on randomly selected datasets  (web-stanford a standard dataset and D70 a synthetic dataset). We apply the static load balancing technique to all parallel variants. With an increase in the number of threads, the speedup rate is significantly less for barrier variants than the \ns variants since each thread has to wait for others in the barrier variants. This also leads \ns variants to have much better scalability in comparison to barrier variants. On the other hand, in \ns variants, as each thread progresses independently, we achieve a higher speedup with a higher thread value. These results suggest, our lock-free variant scales well with the increase in the number of threads.\\
%-----------------------------pgfplot3-------------------
\begin{figure}[!ht]
\vspace*{-3.5cm}
%\begin{adjustbox}{width=110mm, height=100mm}
\centering
\includegraphics[width=6in]{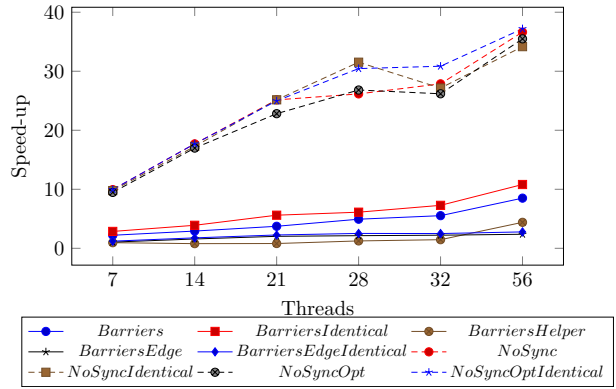}
\DeclareGraphicsExtensions.
\vspace*{-11cm}
\caption{PageRank Speed-Up with Varying threads for web-Stanford Dataset}
\label{fig:lableFig3}
\end{figure}
%--------------------------------------------------

%-----------------------------pgfplot4-------------------
\begin{figure}[!ht]
\vspace*{-3cm}
%\begin{adjustbox}{width=110mm, height=100mm}
\centering
\includegraphics[width=6in]{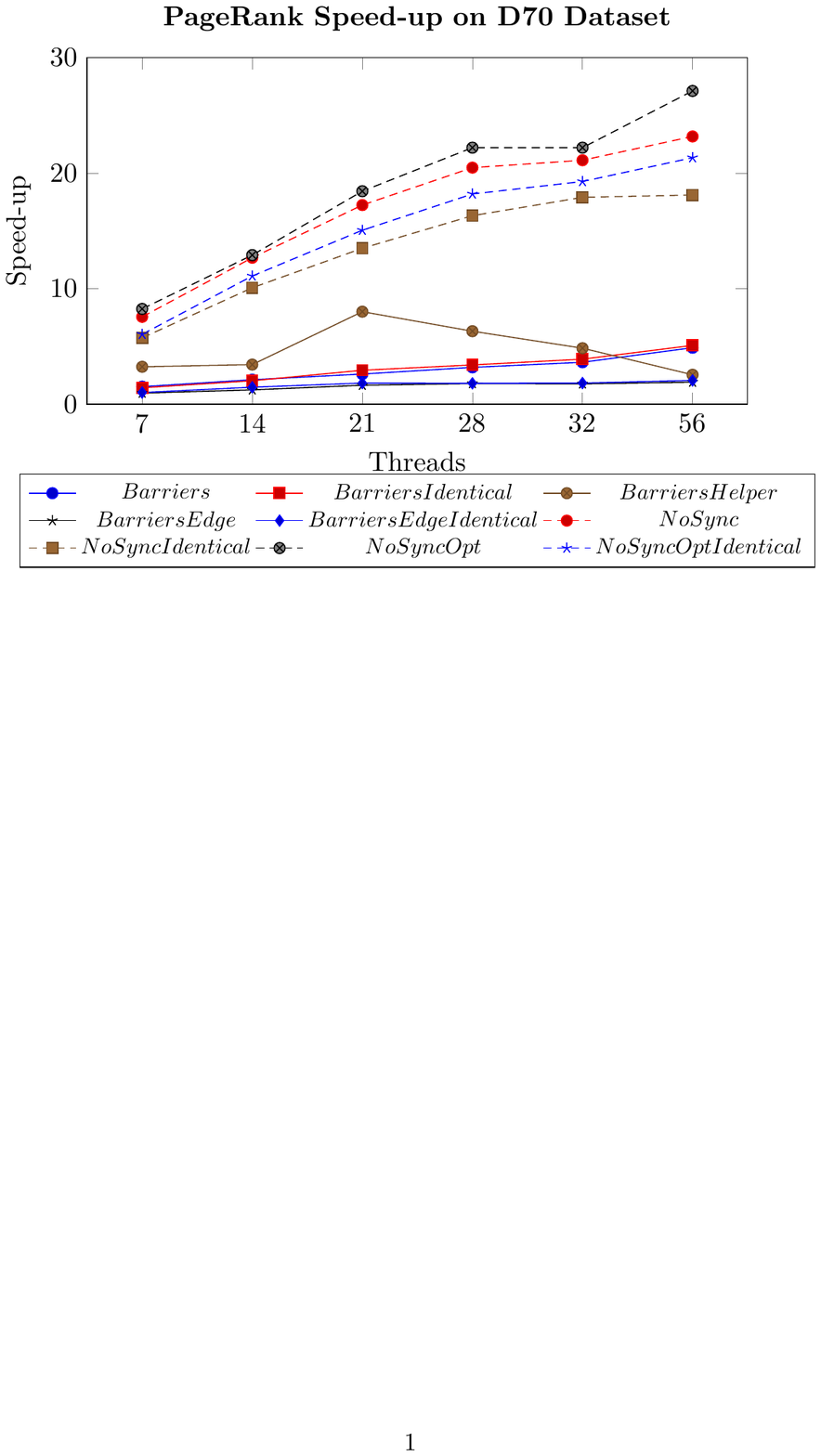}
\DeclareGraphicsExtensions.
\vspace*{-11cm}
\caption{PageRank Speed-Up with Varying threads for D70 Dataset}
\label{fig:lableFig4}
\end{figure}
%--------------------------------------------------

%\hemanote{Description for Plot5 and plot6}
\figref{lableFig5} and \figref{lableFig6} show the speedup and L1-norm obtained by parallel variants on a randomly selected dataset (web-stanford a standard dataset and D70 a synthetic dataset) with a fixed thread count(56). The summation of differences between PageRank of each node from sequential and parallel variants denotes L1-Norm. For most Barrier variants, the L1-norm is zero, which means the page rank values are equal to the sequential ones. No-Sync algorithms, except approximation algorithms on all datasets, is achieving a zero L1-norm. The value is high for \nso and \nsoi as we are performing the loop-perforation technique and skipping the computations when its PageRank value is less than $10^{-21}$. The result of using the above approximation technique leads to an increase in speedup and L1-Norm.\\
%----------------------pgfplot5---------------------------
\begin{figure}[!ht]
\vspace*{-3cm}
%\begin{adjustbox}{width=110mm, height=100mm}
\centering
\includegraphics[width=6in]{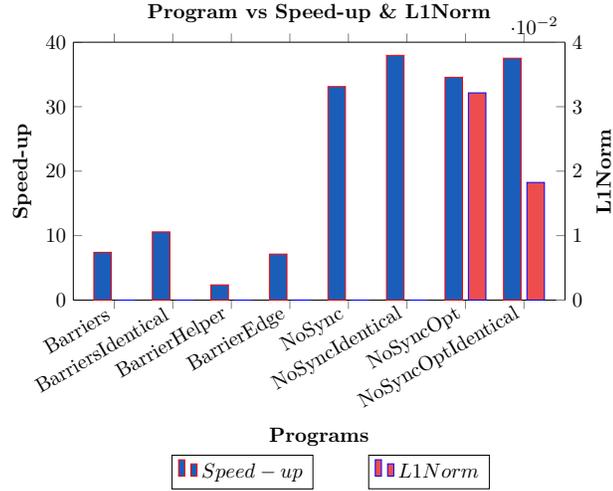}
\DeclareGraphicsExtensions.
\vspace*{-10cm}
\caption{PageRank Speed-Up and L1-Norm for web-standard Dataset)}
\label{fig:lableFig5}
\end{figure} 
%--------------------------------------------------

%----------------------pgfplot6--------------------------

\begin{figure}[!ht]
\vspace*{-3cm}
%\begin{adjustbox}{width=110mm, height=100mm}
\centering
\includegraphics[width=6in]{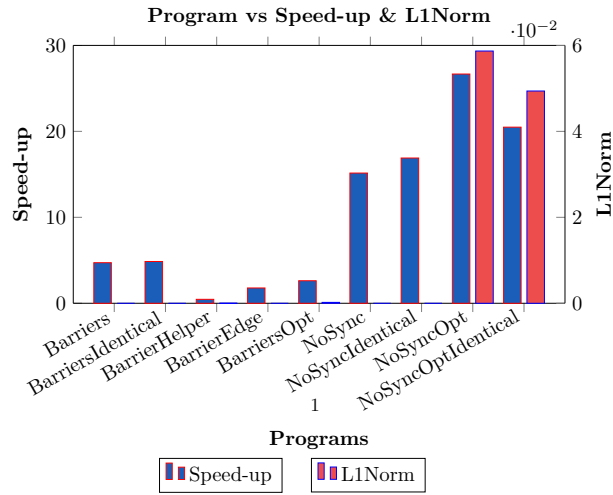}
\DeclareGraphicsExtensions.
\vspace*{-10cm}
\caption{PageRank Speed-Up and L1-Norm for D70 Dataset)}
\label{fig:lableFig6}
\end{figure}
%--------------------------------------------------

%-------------------------------------------------
%\hemanote{Description for Plot7}
In \figref{lableFig7}, we compare the number of iterations taken by each parallel variant. Ideally, we expect each variant to achieve convergence with the same number of iterations. In our case, as we are allowing threads to do partial updates on shared variables that depend on the convergence, \ns variants are taking a fewer number of iterations than barrier variants. Our lock-free variant not only gives better speedup but it also converges faster. Prior to this work, we knew about node-level convergence and algorithm-level convergence on the iterative algorithm, but to our knowledge, we are the first ones to propose thread-level convergence.\\
%------------------pgfplot7--------------------------------
\begin{figure}[!ht]
\vspace*{-3.5cm}
%\begin{adjustbox}{width=110mm, height=100mm}
\centering
\includegraphics[width=6in]{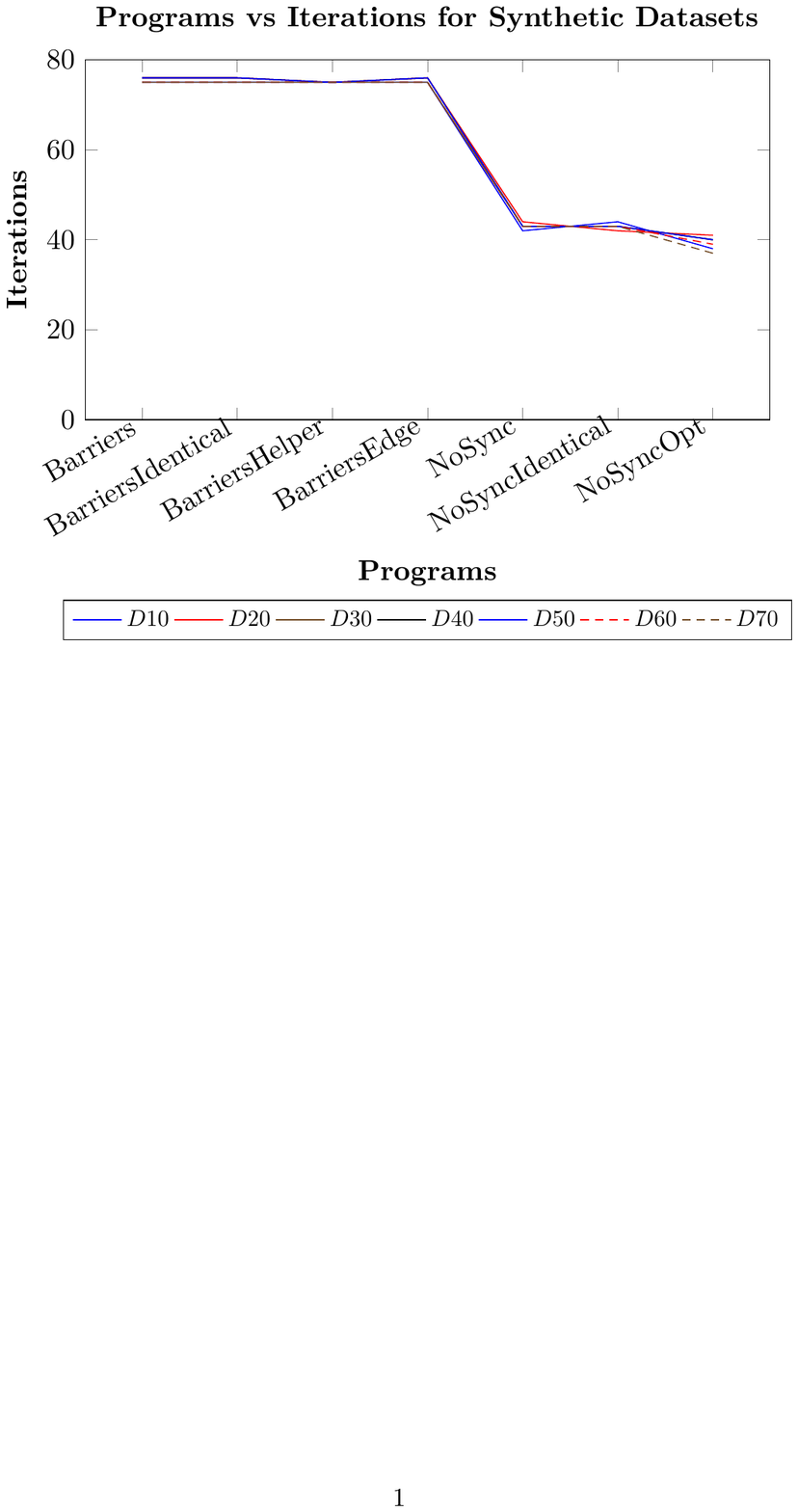}
\DeclareGraphicsExtensions.
\vspace*{-10cm}
\caption{Program Vs \# of Iterations on Synthetic Datasets (56 Threads)}
\label{fig:lableFig7}
\end{figure}
%--------------------------------------------------

%\hemanote{Description for Plot8}
\textbf{Sleeping variants:}  We designed a case study to understand the importance of the \wf algorithm using predetermined steps of calling sleep function to threads during selected iterations. It was observed that every thread waited until the sleeping thread completed its task in the Barrier algorithm, while the task corresponding to the sleeping thread was resumed as soon as the thread woke up in the case of the No-sync approach. Nevertheless, the \wf(Barrier-helper) algorithm was vigorous to address both limitations. The \wf approach is designed so that threads do not have to wait for other threads; instead, they aid other threads in completing their tasks after completing their assigned work. \figref{lableFig8} displays the consistency of the \wf execution time even with an increase in sleep time; however, the election time of No-Sync and Barriers increases as sleep time increases. 

%------------pgfplot8-------------------------------------
\begin{figure}[!ht]
\vspace*{-3.5cm}
%\begin{adjustbox}{width=110mm, height=100mm}
\centering
\includegraphics[width=6in]{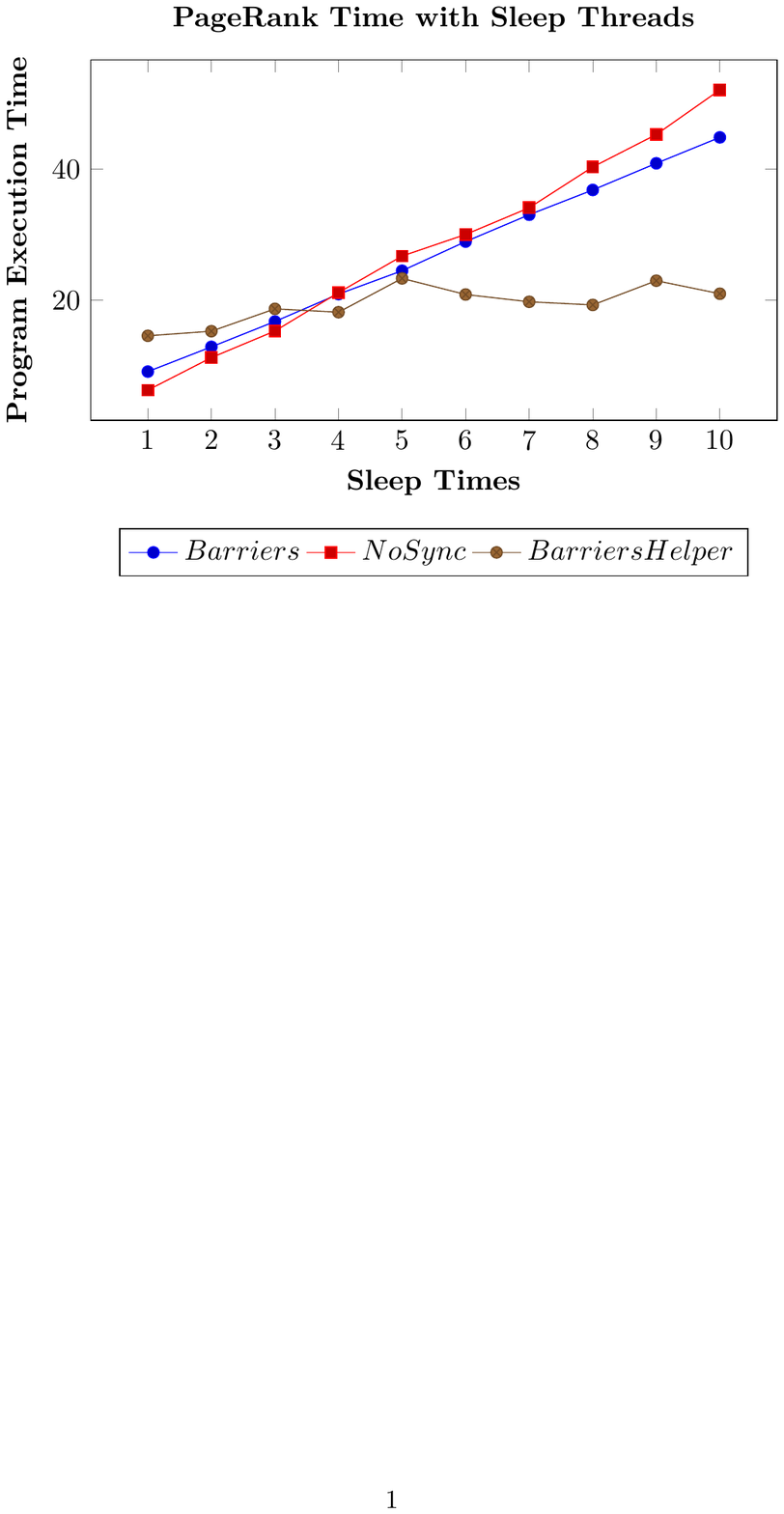}
\DeclareGraphicsExtensions.
\vspace*{-11.5cm}
\caption{Program Execution Time with increase in Sleep Thread}
\label{fig:lableFig8}
\end{figure}
%--------------------------------------------------

%--------------pgfplot9------------------

\begin{figure}[!ht]
\vspace*{-3.5cm}
\centering
\includegraphics[width=6in]{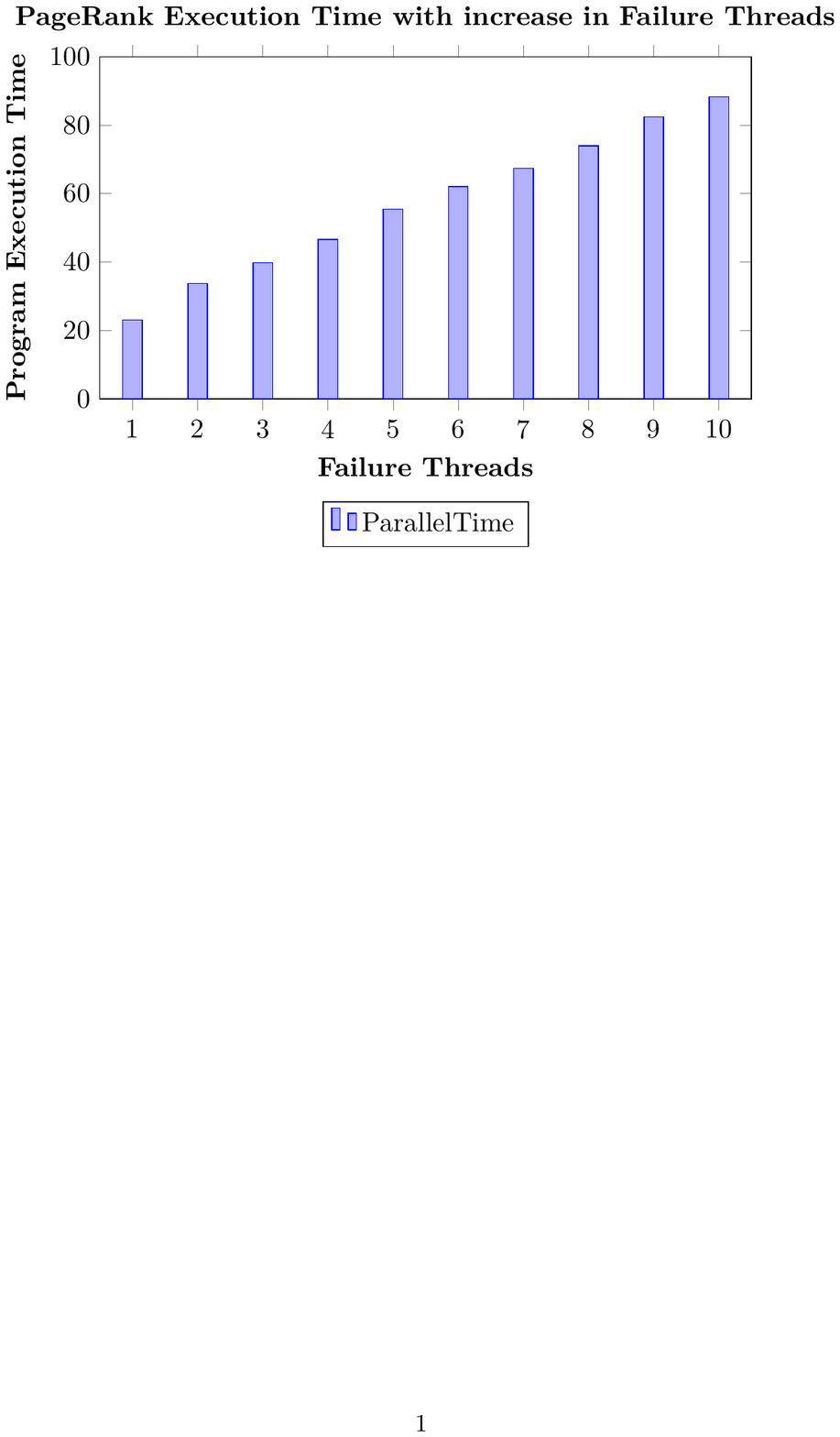}
\DeclareGraphicsExtensions.
\vspace*{-11.5cm}
\caption{Results showing Execution Time for the Failed Threads}
\label{fig:lableFig9}
\end{figure}
%--------------------------------------------------
%\hemanote{Description for Plot9}

\textbf{Failing variants:} \wf algorithm handle thread failures, whereas other parallel algorithms fail to do so. Failures to the threads were added deterministically during the end of the initial iteration to study the effect. As \figref{lableFig9} displays, with the increase in thread failures, the program execution time has increased.
%--------------------normal plot9---------------------- 

%------------------------Experiments & Results Ends--------------------

%----------------------Conclusion and Future Work Starts-----------------
\section{Conclusion and Future work}
The paper proposes a Non-Blocking, i.e., a No-sync and wait-free algorithmic approach for implementing the parallel version of PageRank algorithm on Multiprocessor Systems. This approach has replaced the techniques found in state-of-the-art approaches like Lock-Based and Barrier synchronization algorithms. The results obtained on simulations on various graphs show that the approach proposed in the paper achieves better performance when integrated with the existing methods. The approach with the non-blocking variant applied on iterative algorithms shows performance improvements, as shown by the paper's results.

In the future, the proposed approach will be integrated with existing graph frameworks as an advancement to work. Applying the current proposed approach to the iterative algorithms is also anticipated as further work. 

 The source code is available on \footnote{https://github.com/PDCRL/PageRank}
%----------------------Conclusion and Future Work Ends-----------------

%\label{sec:1}
%Text with citations \cite{RefB} and \cite{RefJ}.
%\subsection{Subsection title}
%\label{sec:2}
%as required. Don't forget to give each section
%and subsection a unique label %(see Sect.~\ref{sec:1}).
%\paragraph{Paragraph headings} Use paragraph headings as needed.
%\begin{equation}
%\end{equation}

%\begin{figure}
  %\includegraphics{example.eps}
%\end{figure}
%\begin{figure}
 % \includegraphics[width=0.75\textwidth]{example.eps}
%\end{figure}

%\input{Table_Results}
%\begin{table}
%%\label{tab:1}       % Give a unique label
%\begin{tabular}{lll}
%\hline\noalign{\smallskip}
%\end{tabular}
%\end{table}

\begin{acknowledgements}
This research is supported by the following grants: (a) "Concurrent and Distributed Programming primitives and algorithms for Temporal Graphs" funded by NSM, GoI. (b) "Tools for Large-Scale Graph Analytics"
funded by Intel, USA.
%We thank the authors of the paper \cite{DBLP:conf/hipc/PanyalaSHKCK17} for sharing the software and helping us port it to our framework.
\end{acknowledgements}
%\hemanote{update}


%\bibliographystyle{spbasic}

%\clearpage
\bibliographystyle{spbasic_unsort.bst}
\bibliography{MainPageRank.bib}

%\begin{thebibliography}{}
%\begin{thebibliography}{}
%\bibitem{Ref9}
%B. W. Barrett, J. W. Berry, R. C. Murphy and K. B. Wheeler, "Implementing a portable Multi-threaded Graph Library: The MTGL on Qthreads," 2009 IEEE International Symposium on Parallel & Distributed Processing, 2009, pp. 1-8, doi: 10.1109/IPDPS.2009.5161102.
% etc
%\end{thebibliography}

% etc
%\end{thebibliography}

\end{document}